\documentclass[runningheads]{llncs}
\usepackage[T1]{fontenc}
\usepackage{graphicx}
\usepackage{orcidlink}
\usepackage{amsfonts} 
\usepackage{bbding}
\usepackage{graphicx}
\usepackage{hyperref}
\usepackage{algorithm}
\usepackage{algpseudocode}
\usepackage{listings}
\usepackage{comment}
\usepackage{tikz}
\usepackage{subfig}
\usepackage{wrapfig}
\usepackage{rotating}
\usepackage{multirow}
\usepackage{comment}
\usepackage{listings}
\usepackage{placeins}
\usepackage{xspace}
\usepackage{siunitx}
\usepackage{stmaryrd}
\usepackage{amsmath} 
\usepackage[capitalize]{cleveref}
\usepackage{enumitem}
\usepackage{ragged2e}
\usepackage{mathtools}
\usepackage{caption}
\usepackage{etoolbox}
\usepackage{booktabs}

\AtBeginEnvironment{algorithmic}{\footnotesize}
\algrenewcommand\algorithmicindent{0.6em}
\algnewcommand\algorithmicinput{\textbf{Input:}}
\algnewcommand\algorithmicoutput{\textbf{Output:}}
\algtext*{EndIf}
\algtext*{EndFor}
\newcommand{\algscale}{0.82}
\newlength{\algwidth}
\newcommand{\st}{\;\ifnum\currentgrouptype=16 \middle\fi|\;}

\newcommand{\rob}{\mu}

\newcommand{\sem}[1]{\llbracket #1 \rrbracket}
\renewcommand{\P}{\mathcal{P}}
\newcommand{\Q}{\mathcal{Q}}

\newcommand{\RR}{\mathbb{R}}

\newcommand{\TT}{\mathbb{T}}

\newcommand{\dom}{{\textit{d}}}

\begin{document}
\title{Quantitative Monitoring of \\ Signal First-Order Logic}
\titlerunning{Quantitative Monitoring of Signal First-Order Logic}
\author{
	Marek Chalupa\inst{1}\orcidlink{0000-0003-1132-5516}, Thomas A. Henzinger\inst{2}\orcidlink{0000-0002-2985-7724}, N. Ege Sara{\c{c}}\inst{3}\orcidlink{0009-0000-2866-8078}, Emily Yu\inst{4}\orcidlink{0000-0002-4993-773X}
}
\authorrunning{M. Chalupa, T.A. Henzinger, N.E. Sara\c{c}, E. Yu}
\institute{
    Zeroth Research, United Kingdom \\ \email{marek.chalupa@zeroth.org} \and
    Institute of Science and Technology Austria, Austria \\ \email{tah@ista.ac.at} \and 
    CISPA Helmholtz Center for Information Security, Germany \\ \email{ege.sarac@cispa.de} \and
	Leiden University, Netherlands \\ \email{z.yu@liacs.leidenuniv.nl}}
\maketitle

\begin{abstract}
Runtime monitoring checks, during execution, whether a partial signal produced by a hybrid system satisfies its specification.
Signal First-Order Logic (SFO) offers expressive real-time specifications over such signals, but currently comes only with Boolean semantics and has no tool support.
We provide the first robustness-based quantitative semantics for SFO, enabling the expression and evaluation of rich real-time properties beyond the scope of existing formalisms such as Signal Temporal Logic.
To enable online monitoring, we identify a past-time fragment of SFO and give a pastification procedure that transforms bounded-response SFO formulas into equisatisfiable formulas in this fragment.
We then develop an efficient runtime monitoring algorithm for this past-time fragment and evaluate its performance on a set of benchmarks, demonstrating the practicality and effectiveness of our approach.
To the best of our knowledge, this is the first publicly available prototype for online quantitative monitoring of full SFO.
\keywords{Signal first-order logic \and Robustness-based quantitative semantics \and Online runtime monitoring.}
\end{abstract}

\section{Introduction}
Runtime verification checks executions of a system against formal specifications by monitoring the system's behavior during runtime.
Online monitors are lightweight pieces of software that operate in parallel with the system, checking partial traces on the fly as the execution unfolds.
In contrast, offline monitoring analyzes the complete trace only after the system has finished executing.

Previous work~\cite{DBLP:conf/emsoft/BakhirkinFHN18} introduced Signal First-Order Logic (SFO) as a rich specification language for expressing real-time properties of real-valued signals.
SFO offers greater expressiveness than the widely used Signal Temporal Logic (STL)~\cite{DBLP:conf/formats/MalerN04}, enabling more precise modeling of complex system behaviors.
For example, SFO allows the formalization of bounded stabilization properties: following a significant disturbance, the signal must return to its pre-disturbance value within 10 timestamps.
This level of expressiveness goes beyond traditional temporal logics and makes SFO well suited for reasoning about the dynamic behavior of hybrid systems.
Technically, SFO combines first-order logic with linear arithmetic and uninterpreted unary function symbols to represent real-valued signals over time.
Its syntax supports quantification over both time and value variables, allowing for rich and flexible property specification.
However, this expressiveness comes at a cost: formulas whose satisfaction depends on unbounded future behavior are in general not efficiently monitorable online, which makes the full logic ill suited for practical runtime monitoring.

So far, the semantics of SFO have been defined only in Boolean terms, i.e., a system trajectory either satisfies or violates a specification~\cite{DBLP:conf/emsoft/BakhirkinFHN18}.
In many real-world applications, especially those involving continuous dynamics and numerical quantities, such binary outcomes are often insufficient.
Quantitative semantics, which measure how well or to what extent a specification is satisfied or violated, offer more informative feedback and can significantly improve decision making in practice.

At a high level, we develop the semantic foundations and algorithmic machinery needed for online quantitative monitoring of SFO.
In particular, we make the following contributions:
\begin{enumerate}
    \item \emph{Quantitative semantics for SFO.}
    We introduce a robust quantitative semantics for SFO, generalizing the standard Boolean semantics and enabling quantitative evaluation of real-valued signals.

    \item \emph{Past-time fragment and pastification.}
    We define a syntactic fragment of SFO, called past-time SFO (pSFO), which is well suited for online monitoring.
    To enable its use in practice, we provide a pastification procedure that transforms bounded-response SFO formulas into equisatisfiable pSFO formulas.

    \item \emph{Online monitoring algorithm.}
    We propose an efficient algorithm for online monitoring of pSFO formulas over piecewise-linear signals.
    Our approach computes robustness values symbolically using polyhedral representations.

    \item \emph{Implementation and experiments.}
    We implement our monitoring algorithm in a prototype tool and evaluate its performance on several benchmark scenarios.
    The results demonstrate the effectiveness of our approach.
\end{enumerate}

\paragraph{Related work.}
Signal Temporal Logic (STL) was introduced as a specification language for continuous signals~\cite{DBLP:conf/formats/MalerN04}, and there is now a large body of work on its quantitative semantics and monitoring algorithms.
Robustness-based semantics for STL~\cite{DBLP:journals/tcs/FainekosP09,DBLP:conf/formats/DonzeM10}, based on signed distance to violation, has since been supported by efficient offline and online monitoring algorithms~\cite{DBLP:conf/cav/DonzeFM13,DBLP:journals/fmsd/DeshmukhDGJJS17,DBLP:conf/rv/DokhanchiHF14,DBLP:journals/sttt/MalerN13,DBLP:journals/fmsd/JaksicBGNN18,DBLP:conf/codit/Gol18}.
Our work is complementary to this STL-based line: we lift robustness-based monitoring from STL to full SFO and identify a past-time, bounded-response fragment that still admits efficient online algorithms.

Beyond STL, the work most closely related to ours is Signal First-Order Logic (SFO)~\cite{DBLP:conf/emsoft/BakhirkinFHN18}:
Bakhirkin et al.\ define a Boolean semantics for SFO and an offline monitoring algorithm, but do not release an implementation or report experimental results.
Menghi et al.~\cite{DBLP:conf/sigsoft/MenghiNGB19} introduce Restricted Signal First-Order Logic (RFOL), a syntactic fragment of SFO with quantitative semantics.
RFOL forbids quantification over value variables, enforces that each subformula contains at most one free time variable, and requires all formulas to be closed.
In contrast, we define a generic robustness-based quantitative semantics for \emph{full} SFO and identify a fragment tailored for efficient online monitoring.

The extension of STL with freeze quantifiers (STL*)~\cite{DBLP:journals/iandc/BrimDSV14} can record the current signal value for later comparison.
While STL* can express relative changes such as local extrema and oscillations, SFO is strictly more expressive: it can capture additional classes of signal-based properties, including punctual derivative constraints that STL* cannot express; see~\cite{DBLP:journals/jss/BoufaiedJBBP21} for a detailed taxonomy and an expressiveness comparison of STL, STL*, and SFO.

In a different direction, Bakhirkin and Basset~\cite{DBLP:conf/tacas/BakhirkinB19} propose a quantitative extension of STL with generalized sliding-window and until operators, tailored to express and monitor bounded-stabilization properties.
However, this extension remains quantifier-free: it cannot bind time or value parameters and reuse them later in the formula, so it cannot capture SFO specifications that quantify over dynamically chosen event durations and relate these durations across different signals, such as rise-time properties (e.g., if signal $f$ has a positive edge then signal $g$ subsequently has a positive edge whose rise time is within $10\%$ of that of $f$; see~\cite[Ex.~3]{DBLP:conf/emsoft/BakhirkinFHN18}).
Silvetti et al.~\cite{DBLP:conf/rv/SilvettiLN25} extend this framework with integral-based and filtered sliding-window operators, which strengthens its ability to express cumulative and conditional quantitative requirements (e.g., average power consumption restricted to periods where a mode is active).
These arithmetic extensions are largely orthogonal: SFO extended with suitable arithmetic operators can encode such sliding-window constructs and, in addition, express duration-parameterized properties that these formalisms cannot capture.

Orthogonal to logic-specific robustness monitoring, recent work develops logic-independent notions of quantitative monitorability, characterizing when (and how well) quantitative trace properties can be monitored (or approximated from finite prefixes)~\cite{DBLP:conf/lics/HenzingerS21,DBLP:conf/rv/HenzingerMS22,DBLP:journals/isse/GorostiagaS25}.
Relatedly, quantitative analogs of the safety-liveness dichotomy and corresponding decompositions have been established for quantitative properties and automata~\cite{DBLP:conf/fossacs/HenzingerMS23,DBLP:conf/concur/BokerHMS23,DBLP:journals/lmcs/BokerHMS25}.
These general frameworks provide semantic criteria for what is monitorable (possibly approximately), whereas we give a robustness semantics for full SFO and exact online monitoring algorithms.

To our knowledge, this paper is the first to provide a generic robustness semantics for full SFO, together with an online monitoring algorithm for a practically relevant past-time bounded-response fragment.

\section{Signal First-Order Logic (SFO)}
In this section, we recall Signal First-Order Logic (SFO) and its standard Boolean semantics~\cite{DBLP:conf/emsoft/BakhirkinFHN18}.
A signal is a function $w : \mathbb{T} \rightarrow \mathbb{R}$ mapping a temporal domain $\mathbb{T} \subseteq \mathbb{R}$ to real values.
Let $\mathcal{F}$ be a set of function symbols and $\mathcal{X}=\mathcal{T}\cup \mathcal{R}$ be a set of variables, partitioned into time variables $\mathcal{T}$ and value variables $\mathcal{R}$.
Both time terms $\tau$ and value terms $\rho$ evaluate to real numbers.
A valuation $\nu$ assigns values to variables $x \in \mathcal{X}$, denoted $\sem{x}_\nu$.
Function symbols remain uninterpreted at the logical level; a signal trace $w$ assigns a real-valued signal $\sem{f}_w$ to each function symbol $f \in \mathcal{F}$.

\begin{definition}[Syntax]
  The syntax of an SFO formula $\phi$ is defined by the following grammar:
  $$\begin{aligned}
    \phi   &::= \theta < \theta \mid \neg \phi \mid \phi \lor \phi \mid \exists r \in R.\,\phi \mid\ \exists s \in I.\,\phi
    \quad& \theta &::= \rho \mid \tau\\
    \rho &::= r \mid f(\tau) \mid n \mid \rho - \rho \mid \rho + \rho
    \quad& \tau &::= s \mid n \mid \tau - \tau \mid \tau + \tau
  \end{aligned}$$
  where $R \subseteq \mathbb{R}$ and $I \subseteq \mathbb{R}$ are intervals with rational end points, $n \in \mathbb{Q}$ is a rational number, $r\in \mathcal{R}$ is a value variable, $s\in \mathcal{T}$ is a time variable, and $\rho$ and $\tau$ denote value and time terms, respectively.
\end{definition}

\begin{definition}[Boolean Semantics]
  Consider an SFO formula $\phi$, a trace $w$, a valuation $\nu$, and a term $\theta$. 
  The satisfaction relation $w,\nu\models \phi$ and the value $\sem{\theta}_{w,\nu}$ are defined inductively as follows:
  $$\begin{array}{l@{\qquad}l}
    \begin{aligned}[t]
      w,\nu \models \theta_1 < \theta_2 &\Leftrightarrow \sem{\theta_1}_{w,\nu} < \sem{\theta_2}_{w,\nu} \\
      w,\nu \models \neg \phi &\Leftrightarrow w,\nu \not\models \phi \\
      w,\nu \models \phi_1 \lor \phi_2 &\Leftrightarrow w,\nu \models \phi_1 \text{ or } w,\nu \models \phi_2 \\
      w,\nu \models \exists r \in R.\,\phi &\Leftrightarrow \exists a \in R: w,\nu{[r \leftarrow a]} \models \phi \\
      w,\nu \models \exists s \in I.\,\phi &\Leftrightarrow \exists a \in I: w,\nu{[s \leftarrow a]} \models \phi
    \end{aligned}
    &
    \begin{aligned}[t]
      \sem{n}_{w,\nu} &= n \quad (n \in \mathbb{Q}) \\
      \sem{x}_{w,\nu} &= \nu(x) \quad (x \in \mathcal{X}) \\
      \sem{\theta_1 \pm \theta_2}_{w,\nu} &= \sem{\theta_1}_{w,\nu} \pm \sem{\theta_2}_{w,\nu} \\
      \sem{f(\tau)}_{w,\nu} &= \sem{f}_w\!\bigl(\sem{\tau}_{w,\nu}\bigr) \\
      \end{aligned}
    \end{array}$$
\end{definition}

We use standard syntactic abbreviations: $\phi \land \psi \equiv \neg(\neg \phi \lor \neg \psi)$, $\phi \rightarrow \psi \equiv \neg \phi \lor \psi$, and $\forall x \in S.\, \phi \equiv \neg \exists x \in S.\, \neg \phi$.
Abbreviations for derived comparison operators ($\leq$, $\geq$, $=$) and absolute values ($|\rho|$) are defined as expected.

\begin{remark}[Well-definedness] \label{rem:well-definedness}
    Formulas are interpreted over traces $w$ with a temporal domain $\mathbb T \subseteq \RR$.
    We only define the satisfaction relation $w,\nu \models \phi$ for pairs $(w,\nu)$ for which all signal accesses in $\phi$ are well-defined:
    whenever a term $f(\tau)$ occurs in $\phi$, we require that the interpreted time $\llbracket \tau \rrbracket_{w,\nu}$ lies in $\mathbb T$.
    If this side condition fails, we leave $w,\nu \models \phi$ undefined.
    In particular, quantifiers range over those instantiations that yield well-defined signal accesses in the quantified subformula; if no such instantiation exists, the formula is undefined.
    Whenever we write $w,\nu \models \phi$, we implicitly assume that this side condition holds.
\end{remark}

At the semantic level, a trace $w$ interprets each function symbol $f$ as a concrete real-valued function~$\llbracket f \rrbracket_w$ on its temporal domain $\mathbb{T}$.
We do not impose any additional assumptions on how these functions are represented or how values between samples are obtained; the semantics depends only on the resulting pointwise function values at the times where the formula accesses the signals.

\section{Quantitative Semantics of SFO}
We introduce a quantitative semantics for SFO, assigning a robustness value to each formula relative to a specific signal and valuation.
This score quantifies the degree to which a signal satisfies or violates a formula.
We further establish the soundness of this semantics with respect to the Boolean semantics.

\begin{definition}[Quantitative Semantics]
    Given an SFO formula $\phi$, a trace $w$, and a valuation $\nu$, the \emph{robustness value} $\rob(w,\nu,\phi) \in \RR \cup \{-\infty, +\infty\}$ is defined inductively as follows.
    \begin{align*}
        \rob(w,\nu,\theta_1<\theta_2)&=\sem{\theta_2}_{w,\nu}-\sem{\theta_1}_{w,\nu}\\
        \rob(w,\nu,\neg \phi)&=-\rob(w,\nu,\phi)\\
        \rob(w,\nu,\phi_1\lor\phi_2)&=\max\{\rob(w,\nu,\phi_1),\rob(w,\nu,\phi_2)\}\\
        \rob(w,\nu, \exists r \in R.\phi)&=\sup_{a\in R} \rob(w,\nu[r\gets a],\phi)\\
        \rob(w,\nu, \exists s \in I.\phi)&=\sup_{a\in I}\rob(w,\nu[s\gets a],\phi)
    \end{align*}
\end{definition}

The quantitative semantics follows the well-definedness convention from~\cref{rem:well-definedness}: $\rob(w,\nu,\phi)$ is defined only when all signal accesses needed to evaluate $\phi$ are defined.
In particular, in the quantified clauses, the supremum ranges only over those instantiations for which the quantified subformula has a defined robustness value; if no such instantiation exists, then $\rob(w,\nu,\phi)$ is undefined.

The definition above is independent of any particular interpolation scheme.
When a trace is given as sampled data, an interpolation choice fixes the concrete functions $\sem{f}_w$ on $\mathbb{T}$.
We illustrate the quantitative semantics of SFO using the bounded stabilization example from the introduction, using piecewise-linear interpolation purely for exposition.

\begin{figure}[t]
\centering
\begin{tikzpicture}[scale=0.7]
  \draw[->] (0,0) -- (15.5,0) node[right] {\scriptsize $t$};
  \draw[->] (0,-0.2) -- (0,3.5) node[above] {\scriptsize $f(t)$};

  \draw[step=1cm, very thin, gray!20] (0,-0.1) grid (15.5,3.1);

  \foreach \x/\lab in {0/0,2.5/5,5/10,7.5/15,10/20,12.5/25,15/30}
    \draw (\x,0) -- (\x,-0.12) node[below] {\small \lab};

    \node[left] at (0,1) {\small 1};
    \node[left] at (0,2) {\small 2};
    \node[left] at (0,3) {\small 3};

  \begin{scope}
    \clip (0.5,-0.2) rectangle (11,3.0);
    \draw[dashed] (0,1.0) -- (5.5,1.0);
    \draw[dotted] (0,1.5) -- (5.5,1.5);
    \draw[dotted] (0,0.5) -- (5.5,0.5);
  \end{scope}

  \begin{scope}
    \clip (6,-0.2) rectangle (15,3.0);
    \draw[dashed] (0,2.375) -- (15,2.375);
    \draw[dotted] (0,2.875) -- (15,2.875);
    \draw[dotted] (0,1.875) -- (15,1.875);
  \end{scope}

  \draw[thick]
    (0, 0.5) -- (0.5,0.5) -- (1, 1.5) -- (1.5, 1) -- (6,1)
    -- (7,3) -- (8,0.5) -- (11,3) -- (15,1);

  \node at (0.5,3) {\tiny $\uparrow$};
  \node[above] at (0.5,3) {\phantom{......}\tiny $t_1=1$};
  \node at (6,3) {\tiny $\uparrow$};
  \node[above] at (5.5,3) {\phantom{......}\tiny $t_2=12$};
\end{tikzpicture}
\caption{Illustration of the quantitative semantics for the bounded stabilization property given in \cref{ex:bstab}. For convenience, we use piecewise-linear interpolation. Upward arrows indicate rising edges at $t_1 = 1$ and $t_2 = 12$. Following $t_1$, the signal $f$ stabilizes rapidly within the $\pm 0.5$ tolerance band around $r_1 = 1$. Conversely, following $t_2$, the signal exhibits wide oscillations that violate the band. The quantitative evaluation considers all $8$-unit windows starting within $10$ units of $t_2$, identifying the window with the minimum range to determine the maximal robustness value.}

\label{fig:bstab-signal}
\end{figure}

\begin{example}[Bounded Stabilization] \label{ex:bstab}
    Let $b : \TT \to \{0,1\}$ be a Boolean mode signal and $f : \TT \to \RR$ be a real-valued signal.
    A \emph{rising edge} of $b$ at time $t$ is
    $$\uparrow b[t] \equiv b(t)=1 \land \exists c \in (0,\infty) . \forall c' \in (0, c) . b(t-c') = 0.$$
    Note that since Boolean signals are piecewise constant and right-continuous, isolated single-point spikes do not arise.
    The \emph{bounded stabilization} requirement states that whenever $b$ rises at $t$, the signal $f$ must enter and remain within a tolerance band $\pm 0.5$ around some value $r$.
    This stabilization must occur no later than $c \leq 10$ time units after $t$ and persist for at least $8$ time units:
    $$\phi \equiv \uparrow b[t] \to \exists r \in \RR . \exists c \in [0,10] . \forall d \in [0,8] . |f(t+c+d)-r| \leq 0.5.$$
    
    Quantitative semantics allow for computing the degree to which a signal and valuation satisfy this property.
    For this example, we consider only time points where a rising edge occurs and focus on the consequent of the implication, denoted by $\psi$.
    Given a signal $w$ and a valuation $\nu$, the quantitative semantics $\rob(w,\nu,\psi)$ is given by:
    $$\sup_{r \in \RR} \sup_{c \in [0,10]} \inf_{d \in [0,8]} (0.5 - |f(t+c+d)-r|).$$
    
    This expression captures the optimal robustness over all candidate stabilization levels $r \in \RR$ and settling delays $c \in [0,10]$.
    Specifically, the inner term represents the worst-case deviation from the $\pm 0.5$ band around $r$ over a window of size $8$ starting at $t+c$.
    Intuitively, for a given $t$, we seek a target value $r$ and a delay $c \in [0,10]$ that minimize the signal's deviation from $r$ in the subsequent 8-unit window.
    The robustness of $\psi$ is upper-bounded by $0.5$ (indicating perfect stabilization) and is unbounded from below.

    Consider the signal in \cref{fig:bstab-signal}.
    After the first rising edge at $t_1=1$, the signal stabilizes perfectly around $r_1=1$ within 2 time units.
    Thus, the parameters $r=1$ and $c=2$ (for any $d \in [0,8]$) maximize the robustness value to $0.5$.
    Conversely, after the second rising edge at $t_2=12$, the signal oscillates widely and slowly.
    To evaluate quantitative satisfaction here, we examine windows of size $8$ starting within $10$ time units of $t_2$.
    For each candidate delay $c \in [0,10]$, we consider the segment $[t_2+c,\,t_2+c+8]$ and compute the range of $f$.
    The robustness value for each segment is $0.5 - \tfrac{1}{2}(\max f - \min f)$, and the overall robustness is the supremum over all feasible $c$ and $r$.
    For the signal in \cref{fig:bstab-signal}, the smallest range of $f$ over any such $8$-unit window occurs in the interval $[19,27]$, which places the peak at $t=22$ in the window while the minima occur at the endpoints.
    Over this interval, the minimum and maximum values of $f$ are $1.75$ and $3.0$, respectively, yielding a range of $1.25$.
    Consequently, the optimal stabilization target is $r_2 = (1.75 + 3.0)/2 = 2.375$, and the corresponding robustness is
    $\rob(w,\nu[t \gets 12],\psi) = 0.5 - \frac{1.25}{2} = -0.125.$
    
    This negative value indicates that no matter how parameters $r$ and $c$ are chosen, the signal fails to remain within the $\pm 0.5$ band around any constant level for a full $8$-unit interval starting within $10$ units of the rising edge.
    Intuitively, the robustness value quantifies the signal's overshoot in the ``best'' window: even in the window with minimal range, the signal deviates by $0.125$ outside the allowed band at its worst point.
    This represents the minimal uniform perturbation required to repair the signal to satisfy the property.
    Shifting the maximum down and the minimum up by $0.125$ compresses the range sufficiently to meet the $\pm 0.5$ requirement.
    Equivalently, the existential quantifiers synthesize parameters: for each fixed $c$, the maximizing $r$ is the midrange of $f$ on $[t+c,t+c+8]$, and the maximizing pair $(r^\star,c^\star)$ witnesses the best robustness value~\cite{DBLP:conf/rv/AsarinDMN11}.
\qed\end{example}
    
\begin{remark}[Units and Compatibility]
	Because atomic predicates are interpreted as differences (i.e., $\theta_2 - \theta_1$), each atomic robustness value carries the units of the compared quantity, and we assume that atomic predicates are dimensionally well-formed.
	When a formula combines predicates over signals with different physical units, the $\max$ aggregator in disjunctions compares incompatible quantities, so the resulting robustness value lacks a direct physical interpretation.
	This limitation is inherited from the standard robustness semantics for STL~\cite{DBLP:journals/tcs/FainekosP09,DBLP:conf/formats/DonzeM10} and is not specific to SFO.
	A practical mitigation is to evaluate all predicates on dimensionless quantities, e.g., by normalizing signals and thresholds using known operating bounds.
	Alternatively, one can adopt an explicitly normalized predicate-level score~\cite{DBLP:conf/sigsoft/MenghiNGB19}.
\end{remark}

Finally, we relate the signed robustness in our quantitative semantics to the standard Boolean semantics.

\begin{theorem}[Soundness of Quantitative Semantics] \label{cl:soundness}
    Let $\phi$ be an SFO formula, $w$ a signal trace, and $\nu$ a valuation.
    If $\rob(w, \nu, \phi) > 0$, then $w, \nu \models \phi$.
    If $\rob(w, \nu, \phi) < 0$, then $w, \nu \not\models \phi$.
\end{theorem}

\section{Past-Time SFO (pSFO)}

Online monitoring is simplest for specifications whose evaluation at time $t$ depends only on the prefix observed up to $t$.
Past-time specifications have this property, so a monitor can emit verdicts (or robustness values) immediately, without buffering future samples.
In this section, we isolate a past-time fragment of SFO and provide a procedure to translate SFO formulas with bounded future lookahead to pSFO formulas.

We call an SFO formula $\phi$ \emph{temporal} if it has exactly one free time variable $t$, and every signal access is anchored at $t$:
whenever $f(\tau)$ occurs in $\phi$, the time term $\tau$ contains $t$ with coefficient $1$ (equivalently, $\tau$ can be rewritten as $t+\delta$ where $t$ does not occur free in $\delta$).
This excludes references to fixed absolute times and ensures that each signal access is at a well-defined offset from the reference time~$t$.

\begin{definition}[Access time terms and pSFO]
	Let $\phi$ be a temporal SFO formula with unique free time variable $t$.
	Let $M_\phi$ be the set of time terms $\tau$ such that some occurrence $f(\tau)$ appears in $\phi$.
	For $\tau \in M_\phi$ and valuation $\nu$, define the \emph{offset} of $\tau$ relative to~$t$ by
	$\mathit{off}_\nu(\tau) = \sem{\tau}_{\nu[t \gets 0]}$.
	We call $\phi$ a \emph{past-time SFO} (pSFO) formula if for every $\tau \in M_\phi$ and every valuation $\nu$ of bound time variables consistent with their quantification intervals (evaluated relative to $t=0$), we have $\mathit{off}_\nu(\tau) \leq 0$.
\end{definition}

Intuitively, $\mathit{off}_\nu(\tau)$ is the time shift encoded by $\tau$ relative to the reference time $t$.
The pSFO condition says that every signal read happens at time $t+\mathit{off}_\nu(\tau) \le t$, i.e., never strictly in the future.

The \emph{bounded-response} fragment of SFO consists of formulas where every time-quantifier interval $I$ in $\exists s \in I.\,\psi$ is bounded, i.e., $\inf I \neq -\infty$ and $\sup I \neq \infty$.
In this fragment, the offsets induced by bound time variables range over bounded sets, so the formula's time dependence can be summarized by finite lookahead and history bounds.

Below, we provide a pastification procedure for temporal bounded-response SFO formulas.
Note that this does not preclude unbounded time entirely: the free time variable $t$ ranges over the entire temporal domain, whereas quantified times act as \emph{relative offsets} and are constrained to bounded ranges.
This allows the expression of safety (and co-safety) properties whose violations (and satisfactions) are witnessed by input segments of bounded length, as is common in practice (e.g., when enforcing deadlines).

The key quantity for pastification is the \emph{forward horizon}, an upper bound on how far into the future the formula may read relative to $t$.
Dually, the \emph{backward horizon} bounds how far into the past it may read.
For bounded-response formulas, both quantities are finite.

\begin{definition}[Horizons]
	Let $\phi$ be a temporal bounded-response formula.
	For each $\tau \in M_\phi$, let $\ell_\tau = \inf_\nu \mathit{off}_\nu(\tau)$ and $u_\tau = \sup_\nu \mathit{off}_\nu(\tau)$, where $\nu$ ranges over valuations of the bound time variables consistent with their quantification intervals.
	We define the \emph{forward} and \emph{backward} horizons, respectively, as $H^+(\phi) = \max_{\tau \in M_\phi} \max(0, u_\tau)$ and $H^-(\phi) = \max_{\tau \in M_\phi} \max(0, -\ell_\tau)$, with the convention that $H^+(\phi)=H^-(\phi)=0$ when $M_\phi=\emptyset$.
\end{definition}

Thus, $H^+(\phi)$ is the maximal future offset of any signal access in $\phi$, and $H^-(\phi)$ is the maximal history length needed to evaluate $\phi$ at a given reference time.
We now define the pastification operator.
Without loss of generality, we assume that bound variables are named apart from free variables.

\begin{definition}[Pastification]
	For $d \geq 0$, the pastification of $\phi$ is $\Pi_d(\phi) = \phi[t \mapsto t-d]$, obtained by substituting every free occurrence of $t$ by $(t-d)$.
\end{definition}

Pastification turns future signal reads into past reads by shifting the reference time backwards inside the formula.
Concretely, if a signal access is at $t+\delta$ in $\phi$, then it becomes $(t-d)+\delta = t + (\delta-d)$ in $\Pi_d(\phi)$; choosing $d$ large enough makes all these new offsets non-positive.

\begin{example}[Pastification for bounded stabilization]\label{ex:past}
	Recall the bounded stabilization property from \cref{ex:bstab}.
	The rising-edge predicate $\uparrow b[t]$ contains an unbounded time quantifier, so $\phi$ is not bounded-response.
	However, the consequent
	$
	\psi \equiv \exists r \in \RR.\ \exists c\in[0,10].\ \forall d\in[0,8].\ |f(t+c+d)-r|\le 0.5
	$
	is bounded-response and can be pastified.
	Here $M_\psi = \{\,t+c+d\,\}$, and for any valuation $\nu$ consistent with $c\in[0,10]$ and $d\in[0,8]$ we have
	$
	\mathit{off}_\nu(t+c+d) = \sem{t+c+d}_{\nu[t\gets 0]} = \nu(c)+\nu(d).
	$
	Hence $\ell_{t+c+d}=0$ and $u_{t+c+d}=18$, so $H^+(\psi)=18$ (and $H^-(\psi)=0$).
	Applying $\Pi_{18}$ yields
	$$
	\Pi_{18}(\psi) \equiv \exists r\in\RR.\ \exists c\in[0,10].\ \forall d\in[0,8].\ |f(t-18+c+d)-r|\le 0.5.
	$$
	Now every access time has offset
	$
	\mathit{off}_\nu(t-18+c+d) = \nu(c)+\nu(d)-18 \le 0,
	$
	so $\Pi_{18}(\psi)$ is a pSFO formula.
	Evaluating $\Pi_{18}(\psi)$ at time $t$ only reads $f$ at times up to $t$, whereas evaluating $\psi$ at time $t$ would need $f$ up to $t+18$.
	Moreover, $\psi$ and $\Pi_{18}(\psi)$ are equisatisfiable via the reference-time shift $t \mapsto t+18$.
	\qed
\end{example}

Choosing $d$ to be at least the forward horizon of the input ensures that all time references lie in the past, i.e., a pSFO formula.
Finally, we establish that pastification preserves satisfiability up to a constant shift of the reference time.

\begin{theorem}[Equisatisfiability under Pastification]
	\label{cl:past-equiv}
	Let $\phi$ be a temporal bounded-response SFO formula with free time variable $t$, and let $h\ge H^+(\phi)$.
	For every trace $w$ and valuation $\nu$, we have $w,\nu \models \phi$ iff $w,\nu' \models \Pi_h(\phi)$, where $\nu'(t)=\nu(t)+h$ and $\nu'(x)=\nu(x)$ for all $x\neq t$.
\end{theorem}

\section{Online Quantitative Monitoring of pSFO}
\label{sec:online}

We present an online quantitative monitoring algorithm for bounded-response pSFO formulas, which capture a broad class of safety and co-safety properties.
For any such formula, the monitor operates over a sliding window corresponding to the formula's backward horizon.
As execution progresses, robustness values are updated incrementally using exclusively the data within this window.
The construction is compositional, enabling the aggregation of monitors for pSFO subformulas while enforcing a bounded memory footprint.
Furthermore, the outputs of these monitors can be combined via Boolean operations, extending the monitoring capability beyond safety and co-safety.

Following the well-definedness convention of~\cref{rem:well-definedness}, the algorithm evaluates robustness only over valuations for which all signal accesses fall within the available trace domain; quantifier ranges that partly exceed the domain are effectively trimmed to the well-defined region.

When $\phi$ is a temporal formula with a unique free time variable $t$,
we write
$\rob(w,t_0,\phi)$ for $\rob\bigl(w,\nu[t\gets t_0],\phi\bigr)$ to denote the \emph{pointwise robustness} of $\phi$ at time $t_0\in\mathbb{T}$,
where $\nu$ is any valuation of the remaining (non-temporal) free variables of $\phi$; if $t$ is the only free variable, the choice of $\nu$ is irrelevant.
Analogously, we write $w,t_0\models\phi$ for $w,\nu[t\gets t_0]\models\phi$.
Since monitoring evaluates a formula at each time point along a trace,
this pointwise view is the one we adopt throughout.

\begin{definition}[Online Monitoring Problem]
	The \emph{online monitoring problem} for a temporal pSFO formula $\phi$ over a stream $w_1, w_2, \ldots$ of input segments requires producing an output segment $v_i$ for each $w_i$.
	In \emph{Boolean monitoring}, $v_i$ is a Boolean segment indicating whether $\phi$ holds pointwise.
	In \emph{quantitative monitoring}, $v_i$ is an extended-real-valued segment providing the pointwise robustness of $\phi$, that is, a partial function $t \mapsto \rob(w,t,\phi) \in \mathbb{R}\cup\{-\infty,+\infty\}$ where defined.
\end{definition}

The core idea of the algorithm is to represent signals, terms, and formulas geometrically as collections of convex polyhedra defined by linear constraints.
Each signal is encoded as a list of polyhedra over time and value variables; arithmetic operations translate into intersections and projections on these polyhedra; and Boolean connectives become polyhedral case-splits.
This geometric perspective, which enables Boolean monitoring~\cite{DBLP:conf/emsoft/BakhirkinFHN18}, forms the foundation for our quantitative extension: by maintaining an auxiliary variable that tracks the robustness of each subformula, we extract quantitative information directly from the polyhedral representation.

Throughout this section, we assume the input signal is sampled at regular time intervals and linearly interpolated between samples.
The assumption of regular sampling is a presentational convenience: the polyhedral operations do not require uniform sampling, and extending to irregular time steps would only affect bookkeeping, not the algorithms themselves.
Although uniform sampling simplifies certain specifications (for instance, property~P2 in the obstacle avoidance experiments in \cref{sec:experiments} approximates a derivative using the fixed step size), SFO and our algorithm could be adapted to monitor such properties without this assumption.

We denote by $\sqcup$, $\sqcap$, $\textsc{Negate}$, and $\textsc{Eliminate}$ the standard operations of union, intersection, complement, and projection, applied to individual polyhedra or lists thereof.
Lists of polyhedra are maintained in increasing order of the upper bound of the free time variable~$t$.
The intersection operation denotes intersections between (i) two polyhedra, (ii) a polyhedron and a list (pairwise), and (iii) two ordered lists of polyhedra.
These operations are efficiently implemented via the double description method~\cite{Motzkin1953DD,DBLP:conf/cococ/FukudaP95}.

Each function $f$ is interpreted as a signal, represented by a list of convex polyhedra $\P_f$ parameterized by two free variables: $t_f$ (time) and $v_f$ (the value of $f$).
These polyhedra are updated upon the arrival of new signal segments; consequently, we treat them as read-only global variables within helper functions.

Finally, we assume that input pSFO formulas are given in a normalized form: all first-order quantifiers are pushed inward as far as possible, and derived propositional connectives are eliminated; negations are pushed inward through propositional structure as far as possible, without pushing them across quantifiers.
This preprocessing preserves the semantics and simplifies the monitoring construction, since quantified variables then appear in local scopes and each quantifier block can be handled independently.
The rewriting is linear in formula size, and pushing quantifiers inward only narrows their scope to the subformulas that mention the bound variable, without duplicating any subformula, so normalization does not increase formula size.

\begin{figure*}[p!]
	\centering
	\scalebox{0.82}{%
		\begin{minipage}[t]{0.6\textwidth}
			\captionsetup{type=algorithm}
			\captionof{algorithm}{\textsc{Monitor}}
			\label{alg:monitoring}
			\parbox{\linewidth}{%
				\begin{algorithmic}[1]
					\Require pSFO formula $\phi$ with backward horizon $h$; sampling period $\Delta$;
					initial signal values $w_0$; stream of signal segments $w_1, w_2, \ldots$
					\Ensure stream of segments $v_1, v_2, \ldots$ encoding the robustness of $w$ for $\phi$
					\vspace{0.4em}
					\State $\P_f \gets \{\text{constraints from } w_0\}$ \text{for all $f \in \mathcal{F}$}
					\For{$i \geq 1$}
					\State Receive $w_i$
					\State $\P_f \gets \P_f \sqcup \{\text{constraints from } w_i\}$ \text{for all $f \in \mathcal{F}$}
					\State $P_\dom \gets \{\, t \in [(i-1)\Delta, i\Delta) \,\}$
					\State $(\P, v) \gets \textsc{FormulaRobust}(\phi, P_\dom)$
					\State Output $v_i \gets (\P, v)$
					\State Drop from $\P_f$ all signal pieces ending before $i\Delta - h$ \text{for all $f \in \mathcal{F}$}
					\EndFor
				\end{algorithmic}
			}%
		\end{minipage}%
	}%
	\hfill
	\scalebox{0.82}{%
		\begin{minipage}[t]{0.6\textwidth}
			\setcounter{algorithm}{3}
			\captionsetup{type=algorithm}
			\captionof{algorithm}{\textsc{EliminateBySup}}
			\label{alg:eliminatebysup}
			\parbox{\linewidth}{%
				\begin{algorithmic}[1]
					\Require list of polyhedra $\P$; current robustness variable $v$; quantified variable $x$; interval $I \subseteq \mathbb{R}$
					\Ensure list of polyhedra $\P'$ (without $x$) and robustness variable $v'$ for the quantified formula on the current segment
					\vspace{0.4em}
					\State $\P' \gets \emptyset$
					\State $v' \gets \textsc{FreshVar}()$
					\ForAll{$P \in \P$}
					\State $P_I \gets P \sqcap \{ x \in I \}$
					\State $({P}_I', o) \gets \textsc{NormalizeObj}(P_I, v, x)$
					\State $(\Q, v') \gets \textsc{PlpMaximize}({P}_I', o, x, v')$
					\State $\P' \gets \P' \sqcup \Q$
					\EndFor
					\State $\P' \gets \{ (Y, v') \in \P' \mid \lnot \exists \hat{v} : (Y, \hat{v}) \in \P' \land \hat{v} > v' \}$
					\State \Return $(\P', v')$
				\end{algorithmic}
			}%
		\end{minipage}%
	}%
	
	\makebox[\textwidth][l]{%
		\scalebox{\algscale}{%
			\begin{minipage}[t]{\algwidth}
				\setcounter{algorithm}{1}
				\captionsetup{type=algorithm}
				\captionof{algorithm}{\textsc{FormulaRobust}}
				\label{alg:formularobust}
				\textbf{Require:} pSFO formula $\phi$; domain-constraints $P_\dom$ for the current segment\\
				\textbf{Ensure:} list of polyhedra $\P$ and robustness variable $v$ for a subformula on the current segment
				
				\vspace{0.4em}
				\begin{minipage}[t]{0.6\textwidth}
					\parbox{\linewidth}{%
						\begin{algorithmic}[1]
							\If{$\phi \equiv (\theta_1 < \theta_2)$}
							\State $(\P', v') \gets \textsc{Term}(\theta_2 - \theta_1, P_\dom)$
							\State \Return $(\P', v')$
							\ElsIf{$\phi \equiv \neg \phi'$}
							\State $(\P', v') \gets \textsc{FormulaRobust}(\phi', P_\dom)$
							\State $v'' \gets \textsc{FreshVar}()$
							\State $\P'' \gets \P' \sqcap \{\, v'' = -v' \,\}$
							\State \Return $(\textsc{Eliminate}(\{v'\}, \P''), v'')$
							\ElsIf{$\phi \equiv \exists r \in R : \phi'$}
							\State $(\P', v') \gets \textsc{FormulaRobust}(\phi', P_\dom \sqcap \{r \in R\})$
							\State $(\P'', v'') \gets \textsc{EliminateBySup}(\P', v', r, R)$
							\State \Return $(\P'', v'')$
							\algstore{frstore}
						\end{algorithmic}
					}%
				\end{minipage}\hspace{-2em}%
				\begin{minipage}[t]{0.6\textwidth}
					\parbox{\linewidth}{%
						\begin{algorithmic}[1]
							\algrestore{frstore}
							\ElsIf{$\phi \equiv \exists s \in I : \phi'$}
							\State $(\P', v') \gets \textsc{FormulaRobust}(\phi', P_\dom \sqcap \{s \in I\})$
							\State $(\P'', v'') \gets \textsc{EliminateBySup}(\P', v', s, I)$
							\State \Return $(\P'', v'')$
							\ElsIf{$\phi \equiv \phi_1 \lor \phi_2$}
							\State $(\P_1, v_1) \gets \textsc{FormulaRobust}(\phi_1, P_\dom)$
							\State $(\P_2, v_2) \gets \textsc{FormulaRobust}(\phi_2, P_\dom)$
							\State $v' \gets \textsc{FreshVar}()$
							\State $\P_{\max}^1 \gets \P_1 \sqcap \P_2 \sqcap \{\, v_1 \geq v_2,\, v' = v_1 \,\}$
							\State $\P_{\max}^2 \gets \P_1 \sqcap \P_2 \sqcap \{\, v_1 <  v_2,\, v' = v_2 \,\}$
							\State $\P' \gets \textsc{Eliminate}\bigl(\{v_1, v_2\},\,
							\P_{\max}^1 \sqcup \P_{\max}^2 \bigr)$
							\State \Return $(\P', v')$
							\EndIf
						\end{algorithmic}
					}%
				\end{minipage}
			\end{minipage}%
		}%
	}
	
	\makebox[\textwidth][l]{%
		\scalebox{\algscale}{%
			\begin{minipage}[t]{\algwidth}
				\setcounter{algorithm}{2}
				\captionsetup{type=algorithm}
				\captionof{algorithm}{\textsc{Term}}
				\label{alg:term}
				\textbf{Require:} a term $\theta$, a polyhedron $P_\dom$ collecting domain constraints\\
				\textbf{Ensure:} list of polyhedra $\P$ and robustness variable $v$ for the term on the current segment
				
				\vspace{0.4em}
				\begin{minipage}[t]{0.6\textwidth}
					\parbox{\linewidth}{%
						\begin{algorithmic}[1]
							\If{$\theta \equiv n$}
							\State $v \gets \textsc{FreshVar}()$
							\State \Return $(\{\, v = n \,\} \sqcap P_\dom, v)$
							\ElsIf{$\theta \equiv r$}
							\State $v \gets \textsc{FreshVar}()$
							\State \Return $(\{\, v = r \,\} \sqcap P_\dom, v)$
							\ElsIf{$\theta \equiv s$}
							\State $v \gets \textsc{FreshVar}()$
							\State \Return $(\{\, v = s \,\} \sqcap P_\dom, v)$
							\algstore{termstore}
						\end{algorithmic}
					}%
				\end{minipage}\hspace{-2em}%
				\begin{minipage}[t]{0.6\textwidth}
					\parbox{\linewidth}{%
						\begin{algorithmic}[1]
							\algrestore{termstore}
							\ElsIf{$\theta \equiv \theta_1 \pm \theta_2$}
							\State $(\P_1, v_1) \gets \textsc{Term}(\theta_1, P_\dom)$
							\State $(\P_2, v_2) \gets \textsc{Term}(\theta_2, P_\dom)$
							\State $v \gets \textsc{FreshVar}()$
							\State $\P \gets (\P_1 \sqcap \P_2) \sqcap \{\, v = v_1 \pm v_2 \,\} \sqcap P_\dom$
							\State \Return $(\textsc{Eliminate}(\{v_1, v_2\}, \P), v)$
							\ElsIf{$\theta \equiv f(\tau)$}
							\State $v \gets \textsc{FreshVar}()$
							\State \Return $(\P_f[t_f \mapsto \tau,\, v_f \mapsto v] \sqcap P_\dom, v)$
							\EndIf
						\end{algorithmic}
					}%
				\end{minipage}
			\end{minipage}%
		}%
	}
	
	\makebox[\textwidth][l]{%
		\scalebox{\algscale}{%
			\begin{minipage}[t]{\algwidth}
				\setcounter{algorithm}{4}
				\captionsetup{type=algorithm}
				\captionof{algorithm}{\textsc{PlpMaximize}}
				\label{alg:parametriclpmaximize}
				\textbf{Require:} polyhedron $P$ over $Y$ and $x$; affine objective $o$ over $(Y,x)$;
				quantified variable $x$; robustness $v'$\\
				\textbf{Ensure:} list of polyhedra $\Q$ over $Y$ with
				$v'(y)=\sup\{o(y,x)\mid (y,x)\in P\}$
				
				\vspace{0.4em}
				\begin{minipage}[t]{0.6\textwidth}
					\parbox{\linewidth}{%
						\begin{algorithmic}[1]
							\State $(\alpha,\beta) \gets \textsc{SplitCoeff}(o,x)$
							\State $(\mathcal{L},\mathcal{U},P_0) \gets \textsc{IsolateBounds}(P,x)$
							\State $P_Y \gets \textsc{Eliminate}(\{x\}, P)$; \quad $\Q \gets \emptyset$
							\State $G^+ \gets P_0 \sqcap \{\alpha > 0\}$
							\State $G^- \gets P_0 \sqcap \{\alpha < 0\}$
							\State $G^0 \gets P_0 \sqcap \{\alpha = 0\}$
							
							\If{$G^+ \neq \emptyset$}
							\If{$\mathcal{U} = \emptyset$}
							\State $\Q \gets \Q \sqcup (G^+ \sqcap P_Y \sqcap \{ v' = +\infty \})$
							\Else
							\ForAll{$u \in \mathcal{U}$}
							\State $A_u \gets \bigsqcap_{u' \in \mathcal{U}} \{ u \leq u' \}$
							\State $F_u \gets \bigsqcap_{\ell \in \mathcal{L}} \{ \ell \leq u \}$
							\State $\Q \gets \Q \sqcup
							(G^+ \sqcap A_u \sqcap F_u \sqcap P_Y \sqcap
							\{ v' = \alpha u + \beta \})$
							\EndFor
							\EndIf
							\EndIf
							\algstore{plpmax}
						\end{algorithmic}
					}%
				\end{minipage}\hspace{-2em}%
				\begin{minipage}[t]{0.6\textwidth}
					\parbox{\linewidth}{%
						\begin{algorithmic}[1]
							\algrestore{plpmax}
							\If{$G^- \neq \emptyset$}
							\If{$\mathcal{L} = \emptyset$}
							\State $\Q \gets \Q \sqcup (G^- \sqcap P_Y \sqcap \{ v' = +\infty \})$
							\Else
							\ForAll{$\ell \in \mathcal{L}$}
							\State $A_\ell \gets \bigsqcap_{\ell' \in \mathcal{L}} \{ \ell \geq \ell' \}$
							\State $F_\ell \gets \bigsqcap_{u \in \mathcal{U}} \{ \ell \leq u \}$
							\State $\Q \gets \Q \sqcup
							(G^- \sqcap A_\ell \sqcap F_\ell \sqcap P_Y \sqcap
							\{ v' = \alpha \ell + \beta \})$
							\EndFor
							\EndIf
							\EndIf
							
							\If{$G^0 \neq \emptyset$}
							\State $\Q \gets \Q \sqcup (G^0 \sqcap P_Y \sqcap \{ v' = \beta \})$
							\EndIf
							
							\State \Return $(\Q, v')$
						\end{algorithmic}
					}%
				\end{minipage}
			\end{minipage}%
		}%
	}
\end{figure*}

\textbf{\textsc{Monitor} (\cref{alg:monitoring})} outlines the high-level procedure.
The monitor maintains a history of signal polyhedra $\P_f$.
Upon receiving segment $w_i$, it appends constraints to $\P_f$, restricts the reference time to the current segment by constructing $P_\dom$, and calls \textsc{FormulaRobust}.
The result is a list of polyhedra $\P$ together with a distinguished robustness variable $v$; this pair encodes the piecewise-affine robustness function on the current segment.
Finally, using the finite backward horizon $h = H^-(\phi)$, the monitor discards signal pieces whose (right) endpoint is strictly smaller than $i\Delta - h$, which are never needed to evaluate $\phi$ for any $t \in [(i-1)\Delta, i\Delta)$.

\textbf{\textsc{FormulaRobust} (\cref{alg:formularobust})} recursively constructs polyhedra for subformulas, mapping logical operators to linear constraints.
\emph{Negation} introduces a fresh output variable and enforces $v''=-v'$, then immediately projects out the old robustness variable; this keeps the representation functional (one designated robustness variable per returned subformula piece).
\emph{Disjunction} encodes the $\max$ operator by case-splitting on $v_1 \ge v_2$ versus $v_1 < v_2$ over the intersection of the two operand regions.
\emph{Existential quantification} computes the supremum of the robustness over the quantified variable's interval via \textsc{EliminateBySup}.

\textbf{\textsc{Term} (\cref{alg:term})} transforms arithmetic terms into polyhedra annotated with a value variable.
In all base cases (constants and variables), the returned polyhedron is conjoined with $P_\dom$ so that any currently active domain constraints (e.g., quantifier bounds) are preserved.
Sums and differences intersect subterm regions and add the defining equality for the output variable, then project away intermediate variables.
For function terms $f(\tau)$, we substitute $\tau$ for the signal time variable $t_f$ and constrain the signal value variable $v_f$ to equal the fresh output variable.

\textbf{\textsc{EliminateBySup} (\cref{alg:eliminatebysup})} eliminates a quantified variable $x$ by taking the supremum of the current robustness variable $v$ over an interval $I$.
It first restricts to the closure $[\inf I,\sup I]$, which does not affect the supremum.
Then, it applies a normalization step: $\textsc{NormalizeObj}(P_I,v,x)$ extracts from $P_I$ the affine expression $o(Y,x)$ that equals $v$ on that piece and substitutes it to eliminate the robustness variable, yielding a reduced polyhedron ${P}_I'$ over $(Y,x)$ together with the explicit objective~$o$.
The reduced polyhedron and explicit objective are then passed to \textsc{PlpMaximize}.
Finally, the last line removes dominated points so that, for each parameter valuation $Y$, only maximal $v'$ values remain; this implements the supremum over a union of pieces as the pointwise maximum of the per-piece suprema.

\textbf{\textsc{PlpMaximize} (\cref{alg:parametriclpmaximize})} performs parametric maximization of an affine objective $o=\alpha x+\beta$ over the feasible $x$-interval induced by a polyhedron $P(Y,x)$.
The optimum is achieved at an extremal bound of $x$ determined by the sign of $\alpha$; if the relevant side is unbounded, the supremum is $+\infty$.
If for some $y$ there is no feasible $x$ (i.e., the section $P_y$ is empty), then $y\notin P_Y$ and the procedure produces no output polyhedron at~$y$, leaving the robustness undefined there (cf. \cref{rem:well-definedness}).

\emph{Lines 1--6.}
$\textsc{SplitCoeff}(o,x)$ decomposes the affine objective into $o=\alpha x+\beta$, where $\alpha$ and $\beta$ are affine in the parameters $Y$.
$\textsc{IsolateBounds}(P,x)$ extracts the affine lower bounds $\mathcal{L}$ ($x \geq \ell(Y)$), upper bounds $\mathcal{U}$ ($x \leq u(Y)$), and the $x$-independent guard $P_0$.
Projecting out $x$ yields $P_Y=\{y \mid \exists x . (y,x) \in P \}$, defining parameter regions where the feasible $x$-interval is nonempty.
The algorithm partitions the parameter space by the sign of $\alpha$:
$G^+$ (positive slope, optimum at an upper bound), $G^-$ (negative slope, optimum at a lower bound), and $G^0$ (zero slope).

\emph{Lines 7--14 ($\alpha > 0$).}
For positive slopes, the supremum occurs at the maximal feasible $x$.
If $\mathcal{U} = \emptyset$, the feasible region is unbounded above and the objective diverges to $+\infty$, yielding $v' = +\infty$ on $G^+ \sqcap P_Y$.
Otherwise, for each candidate $u \in \mathcal{U}$, we isolate the region where $u$ is the tightest upper bound ($A_u$) and feasible w.r.t.\ all lower bounds ($F_u$).
On this region the supremum is $v' = \alpha u + \beta$.

\emph{Lines 15--24 ($\alpha < 0$ and $\alpha = 0$).}
Symmetrically, for negative slopes the supremum occurs at the minimal feasible $x$.
If $\mathcal{L} = \emptyset$, the feasible region is unbounded below and the objective diverges to $+\infty$ as $x \to -\infty$, yielding $v' = +\infty$ on $G^- \sqcap P_Y$.
Otherwise, for each $\ell \in \mathcal{L}$, the region $A_\ell \sqcap F_\ell$ isolates the tightest feasible lower bound, and the supremum is $v' = \alpha \ell + \beta$.
Finally, for $\alpha=0$, the objective is independent of $x$, so $v'=\beta$ on $G^0 \sqcap P_Y$.

\begin{example}[Supremum Elimination]	\label{ex:poly-sup}
	Let $\psi \equiv \exists c \in [0,2] \,.\,  0 < f(t-c)$ and consider the samples $f(0)=0$, $f(1)=-3$, $f(2)=-1$, $f(3)=1$ with piecewise-linear interpolation:
	$f(s)$ is given by $-3s$ for $s \in [0,1)$, and by $2s-5$ for $s \in [1,2)$ and $s \in [2,3)$.
	We encode $f$ as the list $\P_f$ of convex polyhedra over $(t_f,v_f)$:
	$P_f^{(1)} =\{ 0 \leq t_f < 1, v_f = -3t_f \}$, $P_f^{(2)} =\{ 1 \leq t_f < 2, v_f = 2t_f - 5 \}$, and finally $P_f^{(3)} =\{ 2 \leq t_f < 3, v_f = 2t_f - 5 \}$.
	We monitor on the current segment with bounds $ \P_{\dom} = \{ 2 \leq t < 3, 0 \leq c \leq 2 \}. $
	Applying \textsc{Term} to the atom $0<f(t-c)$ substitutes $\tau \coloneqq t-c$ for $t_f$ and introduces a fresh value variable $v$. 
	For each piece of $\P_f$, this produces constraints in $(t,c,v)$, which we intersect with $\P_{\dom}$ and simplify for $t\in[2,3)$:
	\begin{align*}
		\text{(A)}\;\; & 2 \leq t < 3,\ t-1 \leq c \leq 2,\ v = 3c - 3t \\
		\text{(B)}\;\; & 2 \leq t < 3,\ t-2 \leq c \leq t-1,\ v = -2c + 2t - 5,\\
		\text{(C)}\;\; & 2 \leq t < 3,\ 0 \leq c \leq t-2,\ v = -2c + 2t - 5.
	\end{align*}
	Hence, \textsc{FormulaRobust} returns their union with $v$ as the atomic robustness.
	
	The quantifier $\exists c \in [0,2]$ is handled by \textsc{EliminateBySup}, which first conjoins $c\in[0,2]$ and then normalizes each piece:
	$\textsc{NormalizeObj}$ extracts the affine objective $o(t,c)$ (here, simply $o=v$ as given by the displayed equalities) and projects out $v$.
	Then \textsc{PlpMaximize} writes $o=\alpha c+\beta$ (via \textsc{SplitCoeff}) and isolates the affine lower and upper bounds on $c$.
	
	For (A) we have $t-1 \leq c \leq 2$, so $\mathcal{L} = \{t-1\}$, $\mathcal{U} = \{2\}$ and $\alpha = +3 > 0$ (the $G^+$ branch).
	Among the candidate upper bounds, the guard $A_u \land F_u$ with $u = 2$ is the only nonempty one (feasible for all $t \in [2,3)$, since $t-1 \leq 2$), so the supremum is realized at $u=2$, yielding
	$$ v_A'(t) = \alpha u + \beta = 3 \cdot 2 + (-3t) = -3t + 6 \quad (t \in [2,3)). $$
	
	For (B) we have $t-2 \leq c \leq t-1$, so $\mathcal{L} = \{t-2\}$, $\mathcal{U} = \{t-1\}$ and $\alpha = -2 < 0$ (the $G^-$ branch).
	Among the candidate lower bounds, the guard $A_\ell \land F_\ell$ with $\ell = t-2$ is the only nonempty one (feasible for all $t \in [2,3)$, since $t-2 \geq 0$ and $t-2 \leq t-1$), so the supremum is realized at $\ell = t-2$, yielding
	$$ v_B'(t) = \alpha \ell + \beta = -2 (t-2) + (2t - 5) = -1 \quad (t \in [2,3)). $$
	
	For (C) we have $0 \leq c \leq t-2$, so $\mathcal{L} = \{0\}$, $\mathcal{U} = \{t-2\}$ and again $\alpha = -2 < 0$ (the $G^-$ branch).
	Among the candidate lower bounds, the guard $A_\ell \land F_\ell$ with $\ell = 0$ is the only nonempty one (feasible for all $t \in [2,3)$, since $0 \leq t-2$), so the supremum is realized at $\ell = 0$, yielding
	$$ v_C'(t) = \alpha \ell + \beta = -2 \cdot 0 + (2t - 5) = 2t - 5 \quad (t \in [2,3)). $$
	
	Because $t \in [2,3)$ and $c \in [0,2]$, every evaluation time $\tau = t-c$ lies in $[0,3)$, so all three pieces $P_f^{(1)}, P_f^{(2)}, P_f^{(3)}$ are feasible on this segment.
	
	Eliminating $c$ over the union corresponds to taking the pointwise maximum of the three affine outcomes
	$v_A'(t) = -3t + 6$, $v_B'(t) = -1$, and $v_C'(t) = 2t - 5$.
	For all $t \in [2,3)$ we have $v_B'(t) = -1 \leq \max\{v_A'(t), v_C'(t)\}$, so $v_B'$ never attains the supremum.
	Hence the robustness is given by
	$$ g(t) = \rob(w,t,\psi) = \max\{-3t+6,\, 2t-5\}. $$
	The intersection $-3t + 6 = 2t - 5$ occurs at $t = 11/5 = 2.2$, so we obtain
	$$ g(t) =
	\begin{cases}
		-3t+6, & t \in [2, 11/5),\\
		2t-5,  & t \in [11/5, 3).
	\end{cases}$$
	Equivalently, over $(t,v')$, the monitor outputs
	$$
	\{ 2 \leq t < 11/5,\ v' = -3t + 6 \}
	\sqcup
	\{ 11/5 \leq t < 3,\ v' = 2t - 5 \}.
	$$
	
	Note that if we instead quantify $\exists c \in [1,2]$, then for $t\in[0,1)$ we have $\tau=t-c\in[-2,0)$, which lies outside the available trace domain.
	In this case, the polyhedral construction yields an empty projection onto $t$ for that region, and therefore $v'$ is undefined on $[0,1)$ (cf. \cref{rem:well-definedness}).
\qed\end{example}

Finally, we show that our monitoring algorithm is correct by establishing the correctness of each subroutine.
We first prove \textsc{PlpMaximize} correct by characterizing the supremum of an affine function over a polyhedral interval, which is always realized at a boundary determined by affine guards.
This result extends directly to \textsc{EliminateBySup}, which applies the same mechanism to each piece of a polyhedral union and then takes pointwise maxima via dominated-point removal.
We prove afterward that \textsc{Term} is correct by structural induction on terms, since each case preserves the intended affine relation between the output variable and the term's value.
Similarly, we handle \textsc{FormulaRobust} by induction on formulas, showing that each logical operator is correctly mapped to its quantitative counterpart; in particular, disjunction intersects the domains of both operands, enforcing the uniform well-definedness convention (\cref{rem:well-definedness}), while existential quantification ranges only over instantiations for which the robustness of the quantified subformula is defined.
The monitor's correctness follows by composing these results: (i) at each iteration the stored signal pieces cover the backward horizon window $[t-h,t]$ for every $t$ in the current segment, so all subroutine preconditions are met, and (ii) garbage collection preserves this invariant for all future iterations while ensuring bounded memory.

\begin{theorem}[Correctness of \textsc{Monitor} (\Cref{alg:monitoring})] \label{cl:correctness-monitor}
	Let $\phi$ be a pSFO formula with a backward horizon of $h$, and let $w$ be a piecewise-linear signal sampled with period $\Delta$.
	For every $i \geq 1$ and every $t \in [(i-1)\Delta, i\Delta)$ such that $\rob(w,t,\phi)$ is defined, \textsc{Monitor} outputs $v_i$ satisfying $v_i(t) = \rob(w,t,\phi)$.
	For every $t \in [(i-1)\Delta, i\Delta)$ such that $\rob(w,t,\phi)$ is undefined, $v_i(t)$ is undefined.
	Moreover, whenever $v_i(t)$ is defined, it depends only on the restriction of $w$ to $[t-h, t]$.
\end{theorem}

\section{Experiments} \label{sec:experiments}

We implemented our algorithm in Python using \emph{SymPy}~\cite{meurer2017sympy} for symbolic mathematics and \emph{pplpy}~\cite{pplpy_pypi}, a Python wrapper for the \emph{Parma Polyhedra Library}~\cite{bagnara2006parma}.
The implementation\footnote[1]{\url{https://github.com/ista-vamos/artifact-fm-qsfo}} provides an SFO parser, polyhedra representations and operations, and the monitoring algorithm.

Our experiments address a basic feasibility question:
can the proposed online monitor compute robustness values fast enough to keep up with the sampling rate of representative cyber-physical benchmarks?
Concretely, we measure the per-segment runtime of robustness computation and examine how it scales with specification complexity, in particular when moving from horizon-$0$ (history-free) constraints to longer-horizon formulas, and to formulas with nested quantifiers or Boolean combinations of multiple properties.

All experiments were executed on a MacBook Pro M4 (macOS Sequoia 15.7.1) with 48\,GB of memory and 4.5\,GHz CPU.
Although our algorithm is online, experiments were run offline on pre-recorded CSV files for reproducibility.

\subsection{Benchmarks and Properties}

(1) \emph{Dynamic Obstacle Avoidance}.
The first benchmark~\cite{fossen2000survey,qin2021learning} involves a parcel delivery drone flying through an urban environment populated by 8 other drones that act as obstacles. The ego drone is controlled by a neural-network policy, while the remaining drones follow pre-defined trajectories. The environment's state space is eight-dimensional, defined by the vector
$x = [x, y, z, v_x, v_y, v_z, \theta_x, \theta_y]$,
which includes the drone's three spatial coordinates, three velocity components, and two angular variables representing roll and pitch. The control inputs are angular accelerations of $\theta_x$ and $\theta_y$, as well as the vertical thrust $T$.
The control interval is $0.1s$.
We check three different properties in this benchmark:

\begin{enumerate}
    \item[P1:] \emph{Safe velocity}. From time $t_s$ onward, the absolute value of velocity in every direction must be less than $v_{max}m/s$. This property corresponds to the formula
    $t > t_s \implies (|v_x(t)| < v_{max} \land |v_y(t)| < v_{max} \land |v_z(t)| < v_{max})$.
    \item[P2:] \emph{Motion smoothness}. The time-derivative of velocity is bounded: $|\dot{v_x}|\leq a_{max}$, $|\dot{v_y}|\leq a_{max}$,  $|\dot{v_z}|\leq a_{max}$.
    Because we have fixed sampling of $0.1s$, we can approximate $\dot{f}$ as
    $\frac{(f(t) - f(t - 0.1)}{0.1}$ and then rewrite the whole inequality to the expression $|f(t) - f(t - 0.1)| \le c$ where $c = 0.1\cdot a_{max}$. We monitor the conjunction of these inequalities for each derivative above.
    \item[P3:] \emph{Obstacle separation}. For all obstacles, the distance in every direction must be greater than $d_{min}$ at all times.
\end{enumerate}

\noindent(2) \emph{Altitude Control}. 
This benchmark~\cite{heidlauf2018verification} models a simplified F-16 fighter jet that must maintain a safe altitude above the ground. 
The state space is represented by the vector $x = [v_a, \alpha, \beta, \phi, \mathit{pitch}, \mathit{yaw}, \dot{\phi}, \dot{\mathit{pitch}}, \dot{\mathit{yaw}}, n, e, \mathit{alt}, p_{\mathit{lag}}]$
which includes the airspeed ($v_a$), angle of attack ($\alpha$), sideslip angle ($\beta$), roll ($\phi$), pitch, yaw, roll rate, pitch rate, yaw rate, northward and eastward displacements, altitude (\textit{alt}), engine power lag. 
The sampling time of this environment is fixed at 0.033~s. We monitor the following properties:

\begin{enumerate}
\item[P1:] \emph{Safe altitude}. The altitude shall satisfy $300m \le \text{alt}(t) \le 13000m$ at all times.
\item[P2:] \emph{Safe airspeed}. The airspeed shall remain within a safe range $v_{\min} \le v_a(t) \le v_{\max}\; \text{for all } t$.
\item[P3:] \emph{Altitude recovery}. Whenever $\text{alt}(t) < 500$\,m,
then within $\tau_{\text{rec}}$ the altitude shall be at least $700$\,m
and remain so for at least $\tau_{\text{hold}}$.
We monitor the formula $alt(t) < 500 \implies \exists t_r \in [0, \tau_{rec}]. \forall t_h\in[0, \tau_{hold}].\ alt(t + t_r + t_h) \ge 700$.
\end{enumerate}

\subsection{Results and Discussion}

\begin{table}[t]
\caption{The average time in seconds of computing the robustness value for each new segment. }
\renewcommand{\arraystretch}{1.4}
\setlength{\tabcolsep}{5pt} 
\centering
\begin{tabular}{cccc p{1mm} cccc}
\toprule
\multicolumn{4}{c}{\emph{Avoidance}} && \multicolumn{4}{c}{\emph{Altitude Control}} \\
\cmidrule(lr){1-4}
\cmidrule(lr){5-9}
\emph{P1} & \emph{P2} & \emph{P3} & \emph{P1 + P2 + P3} & & \emph{P1} & \emph{P2} & \emph{P3} & \emph{P1 + P2 + P3}  \\
\midrule
0.002 & 0.005 & 0.014 & 0.021 & & 0.002 & 0.002 & 0.210 & 0.532 \\
\bottomrule
\vspace{0.25em}
\end{tabular}
\label{tab:results}
\end{table}

We generated 100 simulation traces from both scenarios with randomized initial states.
For \emph{Dynamic Obstacle Avoidance}, each simulation either stops when the drone collides into nearby drones or reaches a maximum of 1000 simulation steps, whereas the \emph{Altitude Control} simulation has a uniform length of 1000 steps.
The results of our experiments are summarized in Table~\ref{tab:results}.
For the \emph{Dynamic Obstacle Avoidance} scenario, the time of analyzing each observed segment is smaller than the control time ($0.1$\,s) by a margin.
This means that the monitor is timely enough to be used for making control decisions.
One of the reasons for such low computation times is that formulas of properties \emph{P1} and \emph{P3} have horizon 0, which means we do not have to remember and analyze any history.
For property \emph{P2} (and combination of all properties \emph{P1 + P2 + P3}), the horizon is $0.1$ which means remembering one previous segment is sufficient for evaluating the formulas.

For the \emph{Altitude Control} benchmark, first two properties are very similar to the first property of the avoidance benchmark, and therefore the times are similar. Property \emph{P3}, however, is more interesting as the formula contains quantifiers. We used constants $\tau_{rec} = 10$ and $\tau_{hold} = 10$, which means that the formula has horizon $20$.
Since sampling rate is $0.033\,s$, the monitor has to remember more than 600 last segments. Still, it was able to analyze each new segment in around $0.2\,s$. Monitoring all three properties together (\emph{P1 + P2 + P3}) increases the runtime further as the monitor has to compute symbolically minima and maxima over conjunctions and disjunctions.
These computation times are not low enough to be sufficient for control, but are sufficient for raising warnings.

\section{Conclusion}
We introduced the first robustness-based quantitative semantics for Signal First-Order Logic (SFO) and an online monitoring algorithm for a past-time, bounded-response fragment.
Our approach combines pastification with symbolic polyhedral computations to obtain robustness values for expressive SFO specifications over piecewise-linear signals.
A prototype implementation and experiments on representative benchmarks demonstrate that quantitative SFO monitoring is feasible in practice and scales to non-trivial specifications.

Promising directions for future work include using robustness as an objective for controller synthesis, for example as a reward signal in reinforcement learning for cyber-physical systems~\cite{DBLP:journals/ral/MengF23}.
We also aim to develop compositional and decentralized monitoring for multi-agent and distributed settings, where approximate monitoring can improve scalability~\cite{DBLP:conf/rv/BonakdarpourMNS24}.
Another natural extension is integrating SFO monitoring with runtime enforcement~\cite{DBLP:journals/fmsd/FalconeMFR11} (a.k.a.\ shielding~\cite{DBLP:journals/fmsd/KonighoferABHKT17}), where the monitor's low latency relative to the control period could enable a feedback loop that steers the controller away from violations.

\subsubsection*{Acknowledgments}
We thank the anonymous reviewers for their helpful comments.
This work was supported by the European Research Council (ERC) Grants VAMOS (No. 101020093) and HYPER (No. 101055412), and by the Advanced Research and Invention Agency under the Safeguarded AI programme (MSAI-PR01-P047).

\bibliographystyle{plain}
\bibliography{sample-base}

@article{DBLP:journals/lmcs/BokerHMS25,
  author       = {Udi Boker and
                  Thomas A. Henzinger and
                  Nicolas Mazzocchi and
                  N. Ege Sara{\c{c}}},
  title        = {Safety and Liveness of Quantitative Properties and Automata},
  journal      = {Log. Methods Comput. Sci.},
  volume       = {21},
  number       = {2},
  year         = {2025},
  url          = {https://doi.org/10.46298/lmcs-21(2:2)2025},
  doi          = {10.46298/LMCS-21(2:2)2025},
  bibsource    = {dblp computer science bibliography, https://dblp.org}
}

@inproceedings{DBLP:conf/concur/BokerHMS23,
  author       = {Udi Boker and
                  Thomas A. Henzinger and
                  Nicolas Mazzocchi and
                  N. Ege Sara{\c{c}}},
  editor       = {Guillermo A. P{\'{e}}rez and
                  Jean{-}Fran{\c{c}}ois Raskin},
  title        = {Safety and Liveness of Quantitative Automata},
  booktitle    = {34th International Conference on Concurrency Theory, {CONCUR} 2023,
                  Antwerp, Belgium, September 18-23, 2023},
  series       = {LIPIcs},
  volume       = {279},
  pages        = {17:1--17:18},
  publisher    = {Schloss Dagstuhl - Leibniz-Zentrum f{\"{u}}r Informatik},
  year         = {2023},
  url          = {https://doi.org/10.4230/LIPIcs.CONCUR.2023.17},
  doi          = {10.4230/LIPICS.CONCUR.2023.17},
  bibsource    = {dblp computer science bibliography, https://dblp.org}
}

@inproceedings{DBLP:conf/fossacs/HenzingerMS23,
  author       = {Thomas A. Henzinger and
                  Nicolas Mazzocchi and
                  N. Ege Sara{\c{c}}},
  editor       = {Orna Kupferman and
                  Pawel Sobocinski},
  title        = {Quantitative Safety and Liveness},
  booktitle    = {Foundations of Software Science and Computation Structures - 26th
                  International Conference, FoSSaCS 2023, Held as Part of the European
                  Joint Conferences on Theory and Practice of Software, {ETAPS} 2023,
                  Paris, France, April 22-27, 2023, Proceedings},
  series       = {Lecture Notes in Computer Science},
  volume       = {13992},
  pages        = {349--370},
  publisher    = {Springer},
  year         = {2023},
  url          = {https://doi.org/10.1007/978-3-031-30829-1\_17},
  doi          = {10.1007/978-3-031-30829-1\_17},
  bibsource    = {dblp computer science bibliography, https://dblp.org}
}

@inproceedings{DBLP:conf/rv/HenzingerMS22,
  author       = {Thomas A. Henzinger and
                  Nicolas Mazzocchi and
                  N. Ege Sara{\c{c}}},
  editor       = {Thao Dang and
                  Volker Stolz},
  title        = {Abstract Monitors for Quantitative Specifications},
  booktitle    = {Runtime Verification - 22nd International Conference, {RV} 2022, Tbilisi,
                  Georgia, September 28-30, 2022, Proceedings},
  series       = {Lecture Notes in Computer Science},
  volume       = {13498},
  pages        = {200--220},
  publisher    = {Springer},
  year         = {2022},
  url          = {https://doi.org/10.1007/978-3-031-17196-3\_11},
  doi          = {10.1007/978-3-031-17196-3\_11},
  bibsource    = {dblp computer science bibliography, https://dblp.org}
}

@inproceedings{DBLP:conf/lics/HenzingerS21,
  author       = {Thomas A. Henzinger and
                  N. Ege Sara{\c{c}}},
  title        = {Quantitative and Approximate Monitoring},
  booktitle    = {36th Annual {ACM/IEEE} Symposium on Logic in Computer Science, {LICS}
                  2021, Rome, Italy, June 29 - July 2, 2021},
  pages        = {1--14},
  publisher    = {{IEEE}},
  year         = {2021},
  url          = {https://doi.org/10.1109/LICS52264.2021.9470547},
  doi          = {10.1109/LICS52264.2021.9470547},
  bibsource    = {dblp computer science bibliography, https://dblp.org}
}

@article{DBLP:journals/isse/GorostiagaS25,
  author       = {Felipe Gorostiaga and
                  C{\'{e}}sar S{\'{a}}nchez},
  title        = {General monitorability of totally ordered verdict domains},
  journal      = {Innov. Syst. Softw. Eng.},
  volume       = {21},
  number       = {2},
  pages        = {673--686},
  year         = {2025},
  url          = {https://doi.org/10.1007/s11334-024-00557-2},
  doi          = {10.1007/S11334-024-00557-2},
  bibsource    = {dblp computer science bibliography, https://dblp.org}
}

@inproceedings{DBLP:conf/rv/BonakdarpourMNS24,
  author       = {Borzoo Bonakdarpour and
                  Anik Momtaz and
                  Dejan Nickovic and
                  N. Ege Sara{\c{c}}},
  editor       = {Erika {\'{A}}brah{\'{a}}m and
                  Houssam Abbas},
  title        = {Approximate Distributed Monitoring Under Partial Synchrony: Balancing
                  Speed {\&} Accuracy},
  booktitle    = {Runtime Verification - 24th International Conference, {RV} 2024, Istanbul,
                  Turkey, October 15-17, 2024, Proceedings},
  series       = {Lecture Notes in Computer Science},
  volume       = {15191},
  pages        = {282--301},
  publisher    = {Springer},
  year         = {2024},
  url          = {https://doi.org/10.1007/978-3-031-74234-7\_18},
  doi          = {10.1007/978-3-031-74234-7\_18},
  bibsource    = {dblp computer science bibliography, https://dblp.org}
}

@article{DBLP:journals/fmsd/KonighoferABHKT17,
  author       = {Bettina K{\"{o}}nighofer and
                  Mohammed Alshiekh and
                  Roderick Bloem and
                  Laura R. Humphrey and
                  Robert K{\"{o}}nighofer and
                  Ufuk Topcu and
                  Chao Wang},
  title        = {Shield synthesis},
  journal      = {Formal Methods Syst. Des.},
  volume       = {51},
  number       = {2},
  pages        = {332--361},
  year         = {2017},
  url          = {https://doi.org/10.1007/s10703-017-0276-9},
  doi          = {10.1007/S10703-017-0276-9},
  bibsource    = {dblp computer science bibliography, https://dblp.org}
}

@article{DBLP:journals/fmsd/FalconeMFR11,
  author       = {Yli{\`{e}}s Falcone and
                  Laurent Mounier and
                  Jean{-}Claude Fernandez and
                  Jean{-}Luc Richier},
  title        = {Runtime enforcement monitors: composition, synthesis, and enforcement
                  abilities},
  journal      = {Formal Methods Syst. Des.},
  volume       = {38},
  number       = {3},
  pages        = {223--262},
  year         = {2011},
  url          = {https://doi.org/10.1007/s10703-011-0114-4},
  doi          = {10.1007/S10703-011-0114-4},
  bibsource    = {dblp computer science bibliography, https://dblp.org}
}

@article{DBLP:journals/ral/MengF23,
  author       = {Yue Meng and
                  Chuchu Fan},
  title        = {Signal Temporal Logic Neural Predictive Control},
  journal      = {{IEEE} Robotics Autom. Lett.},
  volume       = {8},
  number       = {11},
  pages        = {7719--7726},
  year         = {2023},
  url          = {https://doi.org/10.1109/LRA.2023.3315536},
  doi          = {10.1109/LRA.2023.3315536},
  bibsource    = {dblp computer science bibliography, https://dblp.org}
}

@misc{pplpy_pypi,
  title = {pplpy: Python interface to the Parma Polyhedra Library},
  howpublished = {\url{https://pypi.org/project/pplpy/}},
  note = {Accessed: 2025-12-05}
}

@inproceedings{DBLP:conf/emsoft/BakhirkinFHN18,
  author       = {Alexey Bakhirkin and
                  Thomas Ferr{\`{e}}re and
                  Thomas A. Henzinger and
                  Dejan Nickovic},
  editor       = {Bj{\"{o}}rn B. Brandenburg and
                  Sriram Sankaranarayanan},
  title        = {The first-order logic of signals: keynote},
  booktitle    = {Proceedings of the International Conference on Embedded Software,
                  {EMSOFT} 2018, Torino, Italy, September 30 - October 5, 2018},
  pages        = {1},
  publisher    = {{IEEE}},
  year         = {2018},
  url          = {https://doi.org/10.1109/EMSOFT.2018.8537203},
  doi          = {10.1109/EMSOFT.2018.8537203},
  bibsource    = {dblp computer science bibliography, https://dblp.org}
}

@article{fossen2000survey,
  title={A survey on nonlinear ship control: From theory to practice},
  author={Fossen, Thor I},
  journal={IFAC Proceedings Volumes},
  volume={33},
  number={21},
  pages={1--16},
  year={2000},
  publisher={Elsevier}
}

@article{bagnara2006parma,
  title={The parma polyhedra library},
  author={Bagnara, Roberto and Hill, Patricia M and Zaffanella, Enea and Bagnara, Abramo},
  journal={Sci. Comput. Program},
  volume={72},
  year={2006}
}

@article{meurer2017sympy,
  title={SymPy: symbolic computing in Python},
  author={Meurer, Aaron and Smith, Christopher P and Paprocki, Mateusz and {\v{C}}ert{\'\i}k, Ond{\v{r}}ej and Kirpichev, Sergey B and Rocklin, Matthew and Kumar, AMiT and Ivanov, Sergiu and Moore, Jason K and Singh, Sartaj and others},
  journal={PeerJ Computer Science},
  volume={3},
  pages={e103},
  year={2017},
  publisher={PeerJ Inc.}
}

@article{heidlauf2018verification,
  title={Verification Challenges in F-16 Ground Collision Avoidance and Other Automated Maneuvers.},
  author={Heidlauf, Peter and Collins, Alexander and Bolender, Michael and Bak, Stanley},
  journal={ARCH@ ADHS},
  volume={2018},
  year={2018}
}

@article{qin2021learning,
  title={Learning safe multi-agent control with decentralized neural barrier certificates},
  author={Qin, Zengyi and Zhang, Kaiqing and Chen, Yuxiao and Chen, Jingkai and Fan, Chuchu},
  journal={arXiv preprint arXiv:2101.05436},
  year={2021}
}

@inproceedings{DBLP:conf/rv/AsarinDMN11,
  author       = {Eugene Asarin and
                  Alexandre Donz{\'{e}} and
                  Oded Maler and
                  Dejan Nickovic},
  editor       = {Sarfraz Khurshid and
                  Koushik Sen},
  title        = {Parametric Identification of Temporal Properties},
  booktitle    = {Runtime Verification - Second International Conference, {RV} 2011,
                  San Francisco, CA, USA, September 27-30, 2011, Revised Selected Papers},
  series       = {Lecture Notes in Computer Science},
  volume       = {7186},
  pages        = {147--160},
  publisher    = {Springer},
  year         = {2011},
  url          = {https://doi.org/10.1007/978-3-642-29860-8\_12},
  doi          = {10.1007/978-3-642-29860-8\_12},
  bibsource    = {dblp computer science bibliography, https://dblp.org}
}

@inproceedings{DBLP:conf/codit/Gol18,
  author       = {Ebru Aydin Gol},
  title        = {Efficient Online Monitoring and Formula Synthesis with Past {STL}},
  booktitle    = {5th International Conference on Control, Decision and Information
                  Technologies, CoDIT 2018, Thessaloniki, Greece, April 10-13, 2018},
  pages        = {916--921},
  publisher    = {{IEEE}},
  year         = {2018},
  url          = {https://doi.org/10.1109/CoDIT.2018.8394941},
  doi          = {10.1109/CODIT.2018.8394941},
  bibsource    = {dblp computer science bibliography, https://dblp.org}
}

@inproceedings{DBLP:conf/formats/MalerN04,
  author       = {Oded Maler and
                  Dejan Nickovic},
  editor       = {Yassine Lakhnech and
                  Sergio Yovine},
  title        = {Monitoring Temporal Properties of Continuous Signals},
  booktitle    = {Formal Techniques, Modelling and Analysis of Timed and Fault-Tolerant
                  Systems, Joint International Conferences on Formal Modelling and Analysis
                  of Timed Systems, {FORMATS} 2004 and Formal Techniques in Real-Time
                  and Fault-Tolerant Systems, {FTRTFT} 2004, Grenoble, France, September
                  22-24, 2004, Proceedings},
  series       = {Lecture Notes in Computer Science},
  volume       = {3253},
  pages        = {152--166},
  publisher    = {Springer},
  year         = {2004},
  url          = {https://doi.org/10.1007/978-3-540-30206-3\_12},
  doi          = {10.1007/978-3-540-30206-3\_12},
  bibsource    = {dblp computer science bibliography, https://dblp.org}
}

@inproceedings{DBLP:conf/formats/DonzeM10,
  author       = {Alexandre Donz{\'{e}} and
                  Oded Maler},
  editor       = {Krishnendu Chatterjee and
                  Thomas A. Henzinger},
  title        = {Robust Satisfaction of Temporal Logic over Real-Valued Signals},
  booktitle    = {Formal Modeling and Analysis of Timed Systems - 8th International
                  Conference, {FORMATS} 2010, Klosterneuburg, Austria, September 8-10,
                  2010. Proceedings},
  series       = {Lecture Notes in Computer Science},
  volume       = {6246},
  pages        = {92--106},
  publisher    = {Springer},
  year         = {2010},
  url          = {https://doi.org/10.1007/978-3-642-15297-9\_9},
  doi          = {10.1007/978-3-642-15297-9\_9},
  bibsource    = {dblp computer science bibliography, https://dblp.org}
}

@article{DBLP:journals/tcs/FainekosP09,
  author       = {Georgios E. Fainekos and
                  George J. Pappas},
  title        = {Robustness of temporal logic specifications for continuous-time signals},
  journal      = {Theor. Comput. Sci.},
  volume       = {410},
  number       = {42},
  pages        = {4262--4291},
  year         = {2009},
  url          = {https://doi.org/10.1016/j.tcs.2009.06.021},
  doi          = {10.1016/J.TCS.2009.06.021},
  bibsource    = {dblp computer science bibliography, https://dblp.org}
}

@inproceedings{DBLP:conf/cav/DonzeFM13,
  author       = {Alexandre Donz{\'{e}} and
                  Thomas Ferr{\`{e}}re and
                  Oded Maler},
  editor       = {Natasha Sharygina and
                  Helmut Veith},
  title        = {Efficient Robust Monitoring for {STL}},
  booktitle    = {Computer Aided Verification - 25th International Conference, {CAV}
                  2013, Saint Petersburg, Russia, July 13-19, 2013. Proceedings},
  series       = {Lecture Notes in Computer Science},
  volume       = {8044},
  pages        = {264--279},
  publisher    = {Springer},
  year         = {2013},
  url          = {https://doi.org/10.1007/978-3-642-39799-8\_19},
  doi          = {10.1007/978-3-642-39799-8\_19},
  bibsource    = {dblp computer science bibliography, https://dblp.org}
}

@article{DBLP:journals/fmsd/DeshmukhDGJJS17,
  author       = {Jyotirmoy V. Deshmukh and
                  Alexandre Donz{\'{e}} and
                  Shromona Ghosh and
                  Xiaoqing Jin and
                  Garvit Juniwal and
                  Sanjit A. Seshia},
  title        = {Robust online monitoring of signal temporal logic},
  journal      = {Formal Methods Syst. Des.},
  volume       = {51},
  number       = {1},
  pages        = {5--30},
  year         = {2017},
  url          = {https://doi.org/10.1007/s10703-017-0286-7},
  doi          = {10.1007/S10703-017-0286-7},
  bibsource    = {dblp computer science bibliography, https://dblp.org}
}

@article{DBLP:journals/sttt/MalerN13,
  author       = {Oded Maler and
                  Dejan Nickovic},
  title        = {Monitoring properties of analog and mixed-signal circuits},
  journal      = {Int. J. Softw. Tools Technol. Transf.},
  volume       = {15},
  number       = {3},
  pages        = {247--268},
  year         = {2013},
  url          = {https://doi.org/10.1007/s10009-012-0247-9},
  doi          = {10.1007/S10009-012-0247-9},
  bibsource    = {dblp computer science bibliography, https://dblp.org}
}

@article{DBLP:journals/iandc/BrimDSV14,
  author       = {Lubos Brim and
                  Petr Dluhos and
                  David Safr{\'{a}}nek and
                  Tomas Vejpustek},
  title        = {STL\({}^{\mbox{*}}\): Extending signal temporal logic with signal-value freezing operator},
  journal      = {Inf. Comput.},
  volume       = {236},
  pages        = {52--67},
  year         = {2014},
  url          = {https://doi.org/10.1016/j.ic.2014.01.012},
  doi          = {10.1016/J.IC.2014.01.012},
  bibsource    = {dblp computer science bibliography, https://dblp.org}
}

@inbook{Motzkin1953DD,
url = {https://doi.org/10.1515/9781400881970-004},
title = {3. The Double Description Method},
booktitle = {Contributions to the Theory of Games, Volume II},
author = {T. S. Motzkin and H. Raiffa and G. L. Thompson and R. M. Thrall},
editor = {Harold William Kuhn and Albert William Tucker},
publisher = {Princeton University Press},
address = {Princeton},
pages = {51--74},
doi = {doi:10.1515/9781400881970-004},
isbn = {9781400881970},
year = {1953}
}

@inproceedings{DBLP:conf/cococ/FukudaP95,
  author       = {Komei Fukuda and
                  Alain Prodon},
  editor       = {Michel Deza and
                  Reinhardt Euler and
                  Yannis Manoussakis},
  title        = {Double Description Method Revisited},
  booktitle    = {Combinatorics and Computer Science, 8th Franco-Japanese and 4th Franco-Chinese
                  Conference, Brest, France, July 3-5, 1995, Selected Papers},
  series       = {Lecture Notes in Computer Science},
  volume       = {1120},
  pages        = {91--111},
  publisher    = {Springer},
  year         = {1995},
  url          = {https://doi.org/10.1007/3-540-61576-8\_77},
  doi          = {10.1007/3-540-61576-8\_77},
  bibsource    = {dblp computer science bibliography, https://dblp.org}
}

@article{DBLP:journals/jss/BoufaiedJBBP21,
  author       = {Chaima Boufaied and
                  Maris Jukss and
                  Domenico Bianculli and
                  Lionel Claude Briand and
                  Yago Isasi Parache},
  title        = {Signal-Based Properties of Cyber-Physical Systems: Taxonomy and Logic-based
                  Characterization},
  journal      = {J. Syst. Softw.},
  volume       = {174},
  pages        = {110881},
  year         = {2021},
  url          = {https://doi.org/10.1016/j.jss.2020.110881},
  doi          = {10.1016/J.JSS.2020.110881},
  bibsource    = {dblp computer science bibliography, https://dblp.org}
}

@inproceedings{DBLP:conf/tacas/BakhirkinB19,
  author       = {Alexey Bakhirkin and
                  Nicolas Basset},
  editor       = {Tom{\'{a}}s Vojnar and
                  Lijun Zhang},
  title        = {Specification and Efficient Monitoring Beyond {STL}},
  booktitle    = {Tools and Algorithms for the Construction and Analysis of Systems
                  - 25th International Conference, {TACAS} 2019, Held as Part of the
                  European Joint Conferences on Theory and Practice of Software, {ETAPS}
                  2019, Prague, Czech Republic, April 6-11, 2019, Proceedings, Part
                  {II}},
  series       = {Lecture Notes in Computer Science},
  volume       = {11428},
  pages        = {79--97},
  publisher    = {Springer},
  year         = {2019},
  url          = {https://doi.org/10.1007/978-3-030-17465-1\_5},
  doi          = {10.1007/978-3-030-17465-1\_5},
  bibsource    = {dblp computer science bibliography, https://dblp.org}
}

@inproceedings{DBLP:conf/rv/SilvettiLN25,
  author       = {Simone Silvetti and
                  Michele Loreti and
                  Laura Nenzi},
  editor       = {Bettina K{\"{o}}nighofer and
                  Hazem Torfah},
  title        = {Modular and Online Monitoring of Temporal Logic Specification with
                  Integral and Filter},
  booktitle    = {Runtime Verification - 25th International Conference, {RV} 2025, Graz,
                  Austria, September 15-19, 2025, Proceedings},
  series       = {Lecture Notes in Computer Science},
  volume       = {16087},
  pages        = {120--139},
  publisher    = {Springer},
  year         = {2025},
  url          = {https://doi.org/10.1007/978-3-032-05435-7\_8},
  doi          = {10.1007/978-3-032-05435-7\_8},
  bibsource    = {dblp computer science bibliography, https://dblp.org}
}

@inproceedings{DBLP:conf/sigsoft/MenghiNGB19,
  author       = {Claudio Menghi and
                  Shiva Nejati and
                  Khouloud Gaaloul and
                  Lionel C. Briand},
  title        = {Generating automated and online test oracles for Simulink models with
                  continuous and uncertain behaviors},
  booktitle    = {{ESEC/SIGSOFT} {FSE}},
  pages        = {27--38},
  publisher    = {{ACM}},
  year         = {2019}
}

@inproceedings{DBLP:conf/rv/DokhanchiHF14,
  author       = {Adel Dokhanchi and
                  Bardh Hoxha and
                  Georgios Fainekos},
  title        = {On-Line Monitoring for Temporal Logic Robustness},
  booktitle    = {{RV}},
  series       = {Lecture Notes in Computer Science},
  volume       = {8734},
  pages        = {231--246},
  publisher    = {Springer},
  year         = {2014}
}

@article{DBLP:journals/fmsd/JaksicBGNN18,
  author       = {Stefan Jaksic and
                  Ezio Bartocci and
                  Radu Grosu and
                  Thang Nguyen and
                  Dejan Nickovic},
  title        = {Quantitative monitoring of {STL} with edit distance},
  journal      = {Formal Methods Syst. Des.},
  volume       = {53},
  number       = {1},
  pages        = {83--112},
  year         = {2018}
}

\newpage \appendix
\section{Omitted Proofs}

\subsection{Proof of \cref{cl:soundness} (Soundness of Quantitative Semantics)}

\textbf{Statement.}
Let $\phi$ be an SFO formula, $w$ a signal trace, and $v$ a valuation.
If $\rob(w, v, \phi) > 0$, then $w, v \models \phi$.
If $\rob(w, v, \phi) < 0$, then $w, v \not\models \phi$.

\begin{proof}
	We prove both implications simultaneously by structural induction on the structure of $\phi$.
	
	\smallskip\noindent
	\emph{Atomic case.}
	Let $\phi \equiv (\theta_1 < \theta_2)$.
	By definition, $\rob(w,v,\theta_1<\theta_2)=\sem{\theta_2}_{w,v}-\sem{\theta_1}_{w,v}$.
	If $\rob(w,v,\phi)>0$, then $\sem{\theta_1}_{w,v}<\sem{\theta_2}_{w,v}$, hence $w,v\models \theta_1<\theta_2$.
	If $\rob(w,v,\phi)<0$, then $\sem{\theta_1}_{w,v}>\sem{\theta_2}_{w,v}$, hence $w,v\not\models \theta_1<\theta_2$.
	
	\smallskip\noindent
	\emph{Negation.}
	Let $\phi \equiv \neg \psi$.
	Then $\rob(w,v,\neg\psi)=-\rob(w,v,\psi)$.
	If $\rob(w,v,\neg\psi)>0$, then $\rob(w,v,\psi)<0$ and by the inductive hypothesis we have $w,v\not\models \psi$, hence $w,v\models \neg\psi$.
	If $\rob(w,v,\neg\psi)<0$, then $\rob(w,v,\psi)>0$ and by the inductive hypothesis we have $w,v\models \psi$, hence $w,v\not\models \neg\psi$.
	
	\smallskip\noindent
	\emph{Disjunction.}
	Let $\phi \equiv \phi_1 \lor \phi_2$.
	Then $\rob(w,v,\phi_1\lor\phi_2)=\max\{\rob(w,v,\phi_1),\rob(w,v,\phi_2)\}$.
	If $\rob(w,v,\phi_1\lor\phi_2)>0$, then at least one of $\rob(w,v,\phi_1)$ or $\rob(w,v,\phi_2)$ is strictly greater than 0.
	By the inductive hypothesis, the corresponding disjunct is satisfied, and therefore $w,v\models \phi_1\lor\phi_2$.
	If $\rob(w,v,\phi_1\lor\phi_2)<0$, then both $\rob(w,v,\phi_1)<0$ and $\rob(w,v,\phi_2)<0$.
	By the inductive hypothesis, $w,v\not\models \phi_1$ and $w,v\not\models \phi_2$, hence $w,v\not\models \phi_1\lor\phi_2$.
	
	\smallskip\noindent
	\emph{Existential quantification.}
	We treat $\exists r\in R.\,\psi$ and $\exists s\in I.\,\psi$ uniformly.
	Let $\phi \equiv \exists x\in J.\,\psi$, where $x$ is either a value variable $r$ (with $J=R$) or a time variable $s$ (with $J=I$).
	Let $D \coloneqq \{\,a\in J \mid \rob(w,v[x\gets a],\psi)\ \text{is defined}\,\}$.
	By the well-definedness convention, we have $D\neq\emptyset$, and $\rob(w,v,\exists x\in J.\,\psi)=\sup_{a\in D}\rob(w,v[x\gets a],\psi)$.
	If $\rob(w,v,\exists x\in J.\,\psi)>0$, then there exists some $a\in D$ with $\rob(w,v[x\gets a],\psi)>0$.
	By the inductive hypothesis, $w,v[x\gets a]\models \psi$, hence $w,v\models \exists x\in J.\,\psi$.
	If $\rob(w,v,\exists x\in J.\,\psi)<0$, then for all $a\in D$ we have $\rob(w,v[x\gets a],\psi)<0$.
	By the inductive hypothesis, for all $a\in D$ we have $w,v[x\gets a]\not\models \psi$, hence $w,v\not\models \exists x\in J.\,\psi$.
\qed\end{proof}

\subsection{Proof of \cref{cl:past-equiv} (Equisatisfiability Under Pastification)}

\textbf{Statement.}
Let $\phi$ be a temporal bounded-response SFO formula with free time variable $t$, and let $h\ge H^+(\phi)$.
For every trace $w$ and valuation $v$, we have $w,v \models \phi$ iff $w,v' \models \Pi_h(\phi)$, where $v'(t)=v(t)+h$ and $v'(x)=v(x)$ for all $x\neq t$.

\begin{proof}
	Since $\Pi_h(\cdot)$ is defined by substituting the free variable $t$ with $(t-h)$, we apply the same substitution to terms and write $\Pi_h(\theta)$ for the term obtained from $\theta$ by replacing each occurrence of $t$ by $(t-h)$.
	
	First, we claim that for every (time or value) term $\theta$,
	\begin{equation} \label{eq:past-equiv-term}
		\sem{\Pi_h(\theta)}_{w,v'} \;=\; \sem{\theta}_{w,v}.
	\end{equation}
	
	We prove \eqref{eq:past-equiv-term} by structural induction on $\theta$.
	If $\theta$ is a rational constant $n$, then $\Pi_h(n)=n$ and both sides equal $n$.
	If $\theta$ is a variable $x\neq t$, then $\Pi_h(x)=x$ and $\sem{\Pi_h(x)}_{w,v'}=v'(x)=v(x)=\sem{x}_{w,v}$.
	If $\theta=t$, then $\Pi_h(t)=t-h$ and
	$
	\sem{t-h}_{w,v'} \;=\; v'(t)-h \;=\; (v(t)+h)-h \;=\; v(t) \;=\; \sem{t}_{w,v}.
	$
	If $\theta\equiv \theta_1\pm \theta_2$, then \eqref{eq:past-equiv-term} follows directly from the induction hypothesis and the compositional semantics of $\pm$.
	Finally, if $\theta\equiv f(\tau)$, then by the induction hypothesis applied to the time term $\tau$, we have
	$
	\sem{\Pi_h(f(\tau))}_{w,v'}
	= \sem{f(\Pi_h(\tau))}_{w,v'}
	= \sem{f}_w\!\bigl(\sem{\Pi_h(\tau)}_{w,v'}\bigr)
	= \sem{f}_w\!\bigl(\sem{\tau}_{w,v}\bigr)
	= \sem{f(\tau)}_{w,v}.
	$
	This completes the proof of \eqref{eq:past-equiv-term}.
	As an immediate consequence, for every occurrence of a signal access $f(\tau)$ in $\phi$, the access time $\sem{\tau}_{w,v}$ equals the corresponding access time $\sem{\Pi_h(\tau)}_{w,v'}$ in $\Pi_h(\phi)$; hence the well-definedness condition holds for $(w,v,\phi)$ iff it holds for $(w,v',\Pi_h(\phi))$.
	
	We now prove by structural induction on $\phi$ that
	$w,v\models \phi$ iff $w,v'\models \Pi_h(\phi)$.
		
	\smallskip\noindent
	\emph{Atomic case.}
	Let $\phi \equiv (\theta_1 < \theta_2)$.
	Using \eqref{eq:past-equiv-term} for $\theta_1$ and $\theta_2$,	we have
	$w,v\models \theta_1<\theta_2$
	iff
	$\sem{\theta_1}_{w,v}<\sem{\theta_2}_{w,v}$
	iff 
	$\sem{\Pi_h(\theta_1)}_{w,v'}<\sem{\Pi_h(\theta_2)}_{w,v'}$
	iff
	$w,v'\models \Pi_h(\theta_1<\theta_2)$.
	
	\smallskip\noindent
	\emph{Negation.}
	If $\phi\equiv \neg\psi$, then we have
	$w,v\models \neg\psi$
	iff
	$w,v\not\models \psi$
	iff
	$w,v'\not\models \Pi_h(\psi)$
	iff
	$w,v'\models \neg\Pi_h(\psi)$
	iff
	$w,v'\models \Pi_h(\neg\psi)$,
	where the middle equivalence uses the induction hypothesis for $\psi$.
	
	\smallskip\noindent
	\emph{Disjunction.}
	If $\phi\equiv \psi_1\lor \psi_2$, then by the induction hypothesis,
	$w,v\models \psi_1\lor\psi_2$
	iff
	$(w,v\models \psi_1)\ \text{or}\ (w,v\models \psi_2)$
	iff
	$(w,v'\models \Pi_h(\psi_1))\ \text{or}\ (w,v'\models \Pi_h(\psi_2))$
	iff
	$w,v'\models \Pi_h(\psi_1)\lor \Pi_h(\psi_2)$
	iff
	$w,v'\models \Pi_h(\psi_1\lor\psi_2)$.
	
	\smallskip\noindent
	\emph{Existential quantification over values.}
	If $\phi\equiv \exists r\in R.\,\psi$, then
	$w,v\models \exists r\in R.\,\psi$
	iff
	$\exists a\in R:\ w,v{[r\leftarrow a]}\models \psi$.
	For any $a\in R$, shifting $t$ commutes with updating $r$ (since $r\neq t$), i.e., $(v{[r\leftarrow a]})' = v'{[r\leftarrow a]}$.
	Applying the induction hypothesis to $\psi$ under the valuation $v{[r\leftarrow a]}$ yields
	$w,v{[r\leftarrow a]}\models \psi$
	iff
	$w,v'{[r\leftarrow a]}\models \Pi_h(\psi)$.
	Therefore,
	$w,v\models \exists r\in R.\,\psi$
	iff
	$\exists a\in R:\ w,v'{[r\leftarrow a]}\models \Pi_h(\psi)$
	iff
	$w,v'\models \exists r\in R.\,\Pi_h(\psi)$
	iff
	$w,v'\models \Pi_h(\exists r\in R.\,\psi)$.
	
	\smallskip\noindent
	\emph{Existential quantification over time.}
	If $\phi\equiv \exists s\in I.\,\psi$, the same argument applies: since $s$ is bound, it is not the free variable $t$, and thus $(v{[s\leftarrow a]})' = v'{[s\leftarrow a]}$ for all $a\in I$.
	The induction hypothesis for $\psi$ then yields the desired equivalence.
\qed\end{proof}

\subsection{Proof of \cref{cl:correctness-monitor} (Correctness of \textsc{Monitor} (\Cref{alg:monitoring}))}

\textbf{Statement.}
Let $\phi$ be a pSFO formula with a backward horizon of $h$, and let $w$ be a piecewise-linear signal sampled with period $\Delta$.
For every $i \geq 1$ and every $t \in [(i-1)\Delta, i\Delta)$ such that $\rob(w,t,\phi)$ is defined, \textsc{Monitor} outputs $v_i$ satisfying $v_i(t) = \rob(w,t,\phi)$.
For every $t \in [(i-1)\Delta, i\Delta)$ such that $\rob(w,t,\phi)$ is undefined, $v_i(t)$ is undefined.
Moreover, whenever $v_i(t)$ is defined, it depends only on the restriction of $w$ to $[t-h, t]$.

\medskip
\noindent
\textit{Conventions.}
To avoid confusion between valuations and the robustness variables used in the algorithms, we write $\nu$ for valuations of SFO variables.
For a temporal formula $\phi$ with unique free time variable $t$, we use the standard abbreviation
$\rob(w,t_0,\phi)\coloneqq \rob\bigl(w,\nu[t\gets t_0],\phi\bigr)$.
We also identify a returned pair $(\P,v)$ (a finite union of polyhedra and a distinguished output variable) with the partial function it encodes on the current segment:
for $t_0\in\RR$, define
$$
(\P,v)(t_0) \text{ is defined } \iff \exists a\; \exists P\in\P:\; \nu[t\gets t_0,\,v\gets a]\models P,
$$
$$
(\P,v)(t_0)\coloneqq \sup\{a \mid \exists P\in\P:\; \nu[t\gets t_0,\,v\gets a]\models P\}.
$$
By construction (in particular, the case-splits and the dominated-point removal in \Cref{alg:eliminatebysup}), the set under the supremum is either empty or a singleton, so this agrees with the intended pointwise robustness value whenever defined.

We prove correctness of each subroutine used by \textsc{Monitor} and then compose the results.

\begin{lemma}[Boundary Supremum]\label{lem:boundary-sup}
	Let $P \subseteq \RR^{m}\times\RR$ be a convex polyhedron over parameters $Y \in \RR^{m}$ and a scalar variable $x\in\RR$.
	For $y \in \RR^m$, write
	$$
	P_y \coloneqq \{x \mid (y,x)\in P\}
	\quad\text{and}\quad
	P_Y \coloneqq \{y \mid P_y\neq\emptyset\}.
	$$
	Let $o(y,x)=\alpha(y)\,x+\beta(y)$.
	For every $y\in P_Y$, the set $P_y$ is a (possibly open, possibly unbounded) interval, and we define its endpoints by
	$L(y)\coloneqq \inf P_y\in \RR\cup\{-\infty\}$ and $U(y)\coloneqq \sup P_y\in \RR\cup\{+\infty\}$.
	Then:
	$$
	\sup_{x\in P_y} o(y,x)=
	\begin{cases}
		\alpha(y)\,U(y)+\beta(y) & \text{if }\alpha(y)>0,\\
		\alpha(y)\,L(y)+\beta(y) & \text{if }\alpha(y)<0,\\
		\beta(y) & \text{if }\alpha(y)=0,
	\end{cases}
	$$
	with the conventions that $\alpha(y)\cdot(+\infty)=+\infty$ for $\alpha(y)>0$ and
	$\alpha(y)\cdot(-\infty)=+\infty$ for $\alpha(y)<0$.
\end{lemma}
\begin{proof}
	Fix $y\in P_Y$.
	Since $P$ is convex and defined by finitely many linear inequalities, $P_y$ is the intersection of finitely many halfspaces in~$\RR$ and is therefore a (possibly open, possibly unbounded) interval.
	If $\alpha(y)=0$, then $o(y,x)=\beta(y)$ is constant in~$x$ and the claim follows.
	
	Assume $\alpha(y)>0$.
	Then $o(y,\cdot)$ is strictly increasing.
	If $U(y)=+\infty$, then for every $M$ there exists $x\in P_y$ with $x>M$, hence $o(y,x)$ is arbitrarily large and the supremum is $+\infty$.
	If $U(y)<+\infty$, then for all $x\in P_y$ we have $x\le U(y)$, hence $o(y,x)\le \alpha(y)U(y)+\beta(y)$.
	Conversely, by definition of $U(y)=\sup P_y$, for every $\varepsilon>0$ there exists $x_\varepsilon\in P_y$ with $x_\varepsilon>U(y)-\varepsilon$, so
	$o(y,x_\varepsilon)>\alpha(y)(U(y)-\varepsilon)+\beta(y)$.
	Letting $\varepsilon\to 0$ yields $\sup_{x\in P_y} o(y,x)=\alpha(y)U(y)+\beta(y)$.
	
	The case $\alpha(y)<0$ is symmetric: $o(y,\cdot)$ is strictly decreasing, so the supremum is achieved (as a supremum, not necessarily attained) at the lower endpoint $L(y)$, and if $L(y)=-\infty$ the objective diverges to $+\infty$ as $x\to -\infty$.
\qed\end{proof}

\begin{lemma}[Correctness of \textsc{PlpMaximize} (\Cref{alg:parametriclpmaximize})]\label{lem:plp-correct}
	Let $P$ be a convex polyhedron over variables $(Y,x)$ and let $o(Y,x)$ be an affine objective.
	Let $(\Q,v')$ be the output of $\textsc{PlpMaximize}(P,o,x,v')$.
	Let $\pi_Y(\cdot)$ denote projection onto $Y$-variables.
	We interpret any output piece containing the special constraint $\{v' = +\infty\}$ as assigning the value $+\infty$ (i.e., $v'$ ranges over $\RR\cup\{+\infty\}$ on that piece).
	Then, for every valuation $\nu$ of the parameter variables $Y$,
	$$
	\nu \in \pi_Y\Bigl(\bigcup \Q\Bigr)
	\quad\Longleftrightarrow\quad
	\exists x\in\RR:\; (\nu,x)\in P,
	$$
	and whenever $\nu \in \pi_Y(\bigcup \Q)$, the value $v'(\nu)$ induced by $\Q$ satisfies
	$$
	v'(\nu)=\sup\{\, o(\nu,x)\mid (\nu,x)\in P\,\}.
	$$
	Equivalently, under this convention, $\bigcup \Q$ is the graph of the (extended-real-valued) function $\nu\mapsto \sup\{o(\nu,x)\mid(\nu,x)\in P\}$ on its feasibility domain.
\end{lemma}
\begin{proof}
	Let $o=\alpha x+\beta$ be the decomposition produced by \textsc{SplitCoeff}, where $\alpha$ and $\beta$ are affine in~$Y$.
	The routine \textsc{IsolateBounds} rewrites the constraints of $P$ into the form
	$$
	P \equiv P_0(Y)\ \sqcap\ \Bigl(\bigsqcap_{\ell\in\mathcal{L}}\{x\ \bowtie_\ell\ \ell(Y)\}\Bigr)\ \sqcap\ \Bigl(\bigsqcap_{u\in\mathcal{U}}\{x\ \bowtie_u\ u(Y)\}\Bigr),
	$$
	where $P_0$ contains exactly those constraints of $P$ not involving~$x$, and each $\ell\in\mathcal{L}$ and $u\in\mathcal{U}$ is affine in~$Y$.
	Moreover, each $\bowtie_\ell$ is either $\ge$ or $>$ and each $\bowtie_u$ is either $\le$ or $<$, depending on whether the corresponding original constraint was non-strict or strict.
	Let $P_Y=\textsc{Eliminate}(\{x\},P)$ be the projection onto~$Y$; then $P_Y=\{y \mid P_y\neq\emptyset\}$.
	
	Fix $\nu\in P_Y$.
	Instantiating $Y=\nu$ turns the bounds above into inequalities over~$x$ only, hence $P_\nu$ is an interval with endpoints
	$L(\nu)=\sup\{\ell(\nu)\mid \ell\in\mathcal{L}\}$ (or $-\infty$ if $\mathcal{L}=\emptyset$) and
	$U(\nu)=\inf\{u(\nu)\mid u\in\mathcal{U}\}$ (or $+\infty$ if $\mathcal{U}=\emptyset$).
	The algorithm splits the parameter space by the sign of $\alpha$ into
	$G^+=P_0\sqcap\{\alpha>0\}$, $G^-=P_0\sqcap\{\alpha<0\}$, and $G^0=P_0\sqcap\{\alpha=0\}$,
	and always intersects the produced pieces with $P_Y$, so all output points lie in the feasibility domain.
	
	\emph{Case $\alpha(\nu)>0$.}
	If $\mathcal{U}=\emptyset$ then $U(\nu)=+\infty$ and by \Cref{lem:boundary-sup} the supremum is $+\infty$; the algorithm returns $v'=+\infty$ on $G^+\sqcap P_Y$, hence in particular at~$\nu$.
	If $\mathcal{U}\neq\emptyset$, let $u^\star\in\mathcal{U}$ be an upper bound attaining the minimum at~$\nu$, i.e.\ $u^\star(\nu)=U(\nu)$ (such $u^\star$ exists because $\mathcal{U}$ is finite).
	Then $\nu\models A_{u^\star}$, since $u^\star(\nu)\le u'(\nu)$ for all $u'\in\mathcal{U}$, and $\nu\models F_{u^\star}$, since feasibility implies $\ell(\nu)\le U(\nu)=u^\star(\nu)$ for all $\ell\in\mathcal{L}$.
	Thus $\nu$ lies in the region where the algorithm sets $v'=\alpha u^\star+\beta=\alpha U+\beta$, which equals $\sup_{x\in P_\nu} o(\nu,x)$ by \Cref{lem:boundary-sup}.
	(Ties between several minimizing upper bounds only create overlapping output pieces with the same $v'$ value, which does not affect correctness.)
	
	\emph{Case $\alpha(\nu)<0$.}
	The argument is symmetric.
	If $\mathcal{L}=\emptyset$, then $L(\nu)=-\infty$ and \Cref{lem:boundary-sup} yields supremum $+\infty$; the algorithm outputs $v'=+\infty$ on $G^-\sqcap P_Y$.
	If $\mathcal{L}\neq\emptyset$, pick $\ell^\star\in\mathcal{L}$ attaining the maximum at~$\nu$, so $\ell^\star(\nu)=L(\nu)$; then $\nu\models A_{\ell^\star}$ and $\nu\models F_{\ell^\star}$, and the algorithm sets $v'=\alpha \ell^\star+\beta=\alpha L+\beta=\sup_{x\in P_\nu} o(\nu,x)$.
	
	\emph{Case $\alpha(\nu)=0$.}
	Then $o(\nu,x)=\beta(\nu)$ is independent of~$x$, so $\sup_{x\in P_\nu} o(\nu,x)=\beta(\nu)$ and the algorithm outputs $v'=\beta$ on $G^0\sqcap P_Y$.
	
	Finally, for $\nu\notin P_Y$ we have $P_\nu=\emptyset$, and every output region is intersected with $P_Y$, so $\nu\notin \pi_Y(\bigcup\Q)$.
	This proves both the domain equivalence and the supremum identity.
\qed\end{proof}

\begin{lemma}[Correctness of \textsc{EliminateBySup} (\Cref{alg:eliminatebysup})]
	\label{lem:eliminatebysup-correct}
	Let $\P$ be a finite union of convex polyhedra over variables $(Y,x,v)$, and assume that on each polyhedron $P\in\P$
	the variable $v$ is uniquely determined as an affine function of $(Y,x)$.
	Let $I\subseteq\RR$ be the quantification interval
	and $(\P',v')=\textsc{EliminateBySup}(\P,v,x,I)$.
	Let $\pi_Y(\cdot)$ denote projection onto $Y$-variables.
	Then for every valuation $\nu$ of~$Y$,
	$$
	\nu \in \pi_Y\Bigl(\bigcup \P'\Bigr)
	\quad\Longleftrightarrow\quad
	\exists b\in I:\; \exists P\in\P:\; \exists a:\; \nu[x\gets b,\,v\gets a]\models P,
	$$
	and whenever $\nu \in \pi_Y(\bigcup\P')$, the induced output value satisfies
	$$
	v'(\nu)=\sup\bigl\{\, a \,\big|\,
	\exists b\in I,\ \exists P\in\P:\ \nu[x\gets b,\,v\gets a]\models P \,\bigr\}.
	$$
	In other words, $\bigcup\P'$ is the graph of the (partial) map
	$\nu\mapsto \sup\{\, a \mid \exists b\in I,\ \exists P\in\P:\ \nu[x\gets b,\,v\gets a]\models P \,\}$.
\end{lemma}
\begin{proof}
	For each input polyhedron $P\in\P$, \textsc{EliminateBySup} first forms
	$
	P_I \coloneqq P \sqcap \{x\in I\}.
	$
	Thus, for every valuation $\nu$ of $Y$, the set of feasible instantiations for $(x,v)$ in $P_I$ is exactly
	$
	\bigl\{(b,a)\ \big|\ b\in I \ \land\ \nu[x\gets b,\,v\gets a]\models P \bigr\},
	$
	so both the feasibility domain and the intended supremum over $x\in I$ are preserved.
	
	Next, \textsc{NormalizeObj} eliminates the old robustness variable $v$ by extracting the affine objective
	$o_P(Y,x)$ such that $P_I \models (v=o_P(Y,x))$, producing an equivalent polyhedron $P_I'$ over $(Y,x)$ together with objective $o_P$.
	Therefore, for every $\nu$, we have
	$
	\sup\bigl\{\, a \,\big|\ \exists b\in I:\ \nu[x\gets b,\,v\gets a]\models P \,\bigr\}
	=
	\sup\{\, o_P(\nu,x)\mid (\nu,x)\in P_I'\,\},
	$
	with the left-hand set is empty iff the right-hand feasible set is empty.
	
	Applying \textsc{PlpMaximize} to $(P_I',o_P)$ yields a union $\Q_P$ over $(Y,v')$ such that, by \Cref{lem:plp-correct},
	for all $\nu$, if $\nu \in \pi_Y\Bigl(\bigcup\Q_P\Bigr)$ then $v'(\nu)=\sup\{\, o_P(\nu,x)\mid (\nu,x)\in P_I'\,\}$.
	Moreover, \textsc{PlpMaximize} intersects every generated guard with
	$P_Y \coloneqq \textsc{Eliminate}(\{x\},P_I')$, hence
	$\pi_Y(\bigcup\Q_P)\subseteq P_Y$.
	Conversely, if $\nu\in P_Y$ then there exists $x$ with $(\nu,x)\in P_I'$.
	Writing $o_P=\alpha x+\beta$, the case split $\alpha>0$, $\alpha<0$, $\alpha=0$ places $\nu$ in one of the corresponding regions
	$G^+$, $G^-$, $G^0$ of \Cref{alg:parametriclpmaximize}.
	In each case, either the relevant bound set is empty (and the algorithm adds the whole region intersected with $P_Y$),
	or a tightest bound is attained (the bound sets are finite), so $\nu$ satisfies the corresponding guard $A_u\sqcap F_u$
	(or $A_\ell\sqcap F_\ell$), and hence $\nu\in \pi_Y(\bigcup\Q_P)$.
	Therefore $\pi_Y(\bigcup\Q_P)=P_Y$, i.e., $\nu \in \pi_Y\Bigl(\bigcup\Q_P\Bigr)$ iff $\exists x:\ (\nu,x)\in P_I'$ iff $\exists b\in I:\ \exists a:\ \nu[x\gets b,\,v\gets a]\models P$.
	The algorithm then takes the union $\bigcup_{P\in\P}\Q_P$ over all pieces.
	Fix $\nu$ such that $\exists b\in I,\ \exists P\in\P,\ \exists a:\ \nu[x\gets b,\,v\gets a]\models P$.
	Since a supremum over a finite union of feasible sets equals the maximum of the per-piece suprema, we have
	\begin{align*}
		&\sup\bigl\{\, a \,\big|\ \exists b\in I,\ \exists P\in\P:\ \nu[x\gets b,\,v\gets a]\models P \,\bigr\} \\
		&\;=\;
		\max_{P\in\P:\ \exists b\in I\ \exists a:\ \nu[x\gets b,\,v\gets a]\models P}\ \sup\bigl\{\, a \,\big|\ \exists b\in I:\ \nu[x\gets b,\,v\gets a]\models P \,\bigr\}.
	\end{align*}
	The final dominated-point removal step in \Cref{alg:eliminatebysup} keeps, for each fixed $\nu$, exactly the points with maximal $v'$ among all pieces, thereby implementing this pointwise maximum.
	In particular, it preserves the domain $\pi_Y(\bigcup\P')$ (it only removes non-maximal values at already-present $\nu$).
	Thus $\bigcup\P'$ is precisely the graph of the supremum function on its feasibility domain, and it is empty over those $\nu$ for which no feasible instantiation $b\in I$ exists.
\qed\end{proof}

\begin{lemma}[Correctness of \textsc{Term} (\Cref{alg:term})]\label{lem:term-correct}
	Fix a trace $w$.
	Let $\theta$ be a term and let $P_\dom$ be the current conjunction of domain constraints passed to \textsc{Term}.
	Let $(\P,v)=\textsc{Term}(\theta,P_\dom)$.
	Then for every valuation $\nu$ of the variables occurring in $\theta$ or in $P_\dom$ (excluding the fresh output variable $v$),
	$$
	\Bigl(\nu\models P_\dom \ \land\ \sem{\theta}_{w,\nu}\ \text{is defined}\Bigr)
	\quad\Longleftrightarrow\quad
	\exists a:\ \nu[v\gets a]\models \bigcup\P,
	$$
	and whenever this holds, the (necessarily unique) witnessing $a$ equals $\sem{\theta}_{w,\nu}$.
\end{lemma}
\begin{proof}
	Throughout the proof, we use that for each function symbol $f$ the current list $\P_f$ encodes the graph of	$\llbracket f \rrbracket_w$ on the time domain currently stored by the monitor (with piecewise-linear interpolation).
	We proceed by structural induction on $\theta$.
	
	If $\theta$ is a constant $n$ or a variable ($r$ or $s$), \textsc{Term} returns a single polyhedron enforcing $v=n$ (resp.\ $v=r$, $v=s$) conjoined with $P_\dom$, so the claim is immediate.
	
	If $\theta\equiv \theta_1\pm \theta_2$, let $(\P_1,v_1)$ and $(\P_2,v_2)$ be the recursive results.
	By the induction hypothesis, for any $\nu\models P_\dom$, satisfying $\bigcup\P_1$ (resp.\ $\bigcup\P_2$) is equivalent to $\sem{\theta_1}_{w,\nu}$ (resp.\ $\sem{\theta_2}_{w,\nu}$) being defined, and in that case the unique values of $v_1$ and $v_2$	equal these semantics.
	The algorithm constructs $\P$ by intersecting $\P_1$ and $\P_2$, adding the constraint $v=v_1\pm v_2$, and projecting away $v_1$ and $v_2$.
	Therefore $\nu[v\gets a]\models\bigcup\P$ holds exactly when both subterms are defined and	$a=\sem{\theta_1}_{w,\nu}\pm \sem{\theta_2}_{w,\nu}=\sem{\theta}_{w,\nu}$.
	
	If $\theta\equiv f(\tau)$, then by assumption $\P_f$ is a polyhedral graph representation of $f$:
	for any time point $s$ in the currently stored signal domain there is a polyhedron in $\P_f$ enforcing $v_f=\llbracket f\rrbracket_w(s)$ when $t_f=s$.
	\textsc{Term} substitutes $t_f\mapsto \tau$ and $v_f\mapsto v$ and conjoins $P_\dom$.
	Hence, for any valuation $\nu\models P_\dom$, the constraints are satisfiable exactly when the access time $\sem{\tau}_{w,\nu}$ lies in the stored domain of $f$ (i.e., the access is well-defined with respect to the available trace), and then they enforce $v=\llbracket f\rrbracket_w(\sem{\tau}_{w,\nu})=\sem{f(\tau)}_{w,\nu}$.
\qed\end{proof}

\begin{lemma}[Correctness of \textsc{FormulaRobust} (\Cref{alg:formularobust})]\label{lem:fr-correct}
	Let $\phi$ be a pSFO formula and let $P_\dom$ be the current conjunction of domain constraints for the free variables of~$\phi$ on the current segment.
	Let $(\P,v)=\textsc{FormulaRobust}(\phi,P_\dom)$.
	Then for every valuation $\nu$ of the free variables of $\phi$ (excluding the fresh output variable $v$),
	$$
	\Bigl(\nu\models P_\dom \ \land\ \rob(w,\nu,\phi)\ \text{is defined}\Bigr)
	\quad\Longleftrightarrow\quad
	\exists a\in \RR\cup\{-\infty,+\infty\}:\ \nu[v\gets a]\models \bigcup\P,
	$$
	and whenever this holds, the witnessing $a$ is unique and equals $\rob(w,\nu,\phi)$.
\end{lemma}
\begin{proof}
	We prove the claim by structural induction on $\phi$.
	
	Let $\phi\equiv(\theta_1<\theta_2)$.
	Then \textsc{FormulaRobust} calls \textsc{Term} on the term $\theta_2-\theta_1$ and returns the resulting pair $(\P,v)$.
	By \Cref{lem:term-correct}, for every valuation $\nu$ we have that
	$\nu\models P_\dom$ and $\sem{\theta_2-\theta_1}_{w,\nu}$ is defined iff there exists a (necessarily unique) $a\in\RR$ with
	$\nu[v\gets a]\models\bigcup\P$, and then $a=\sem{\theta_2-\theta_1}_{w,\nu}$.
	The claim follows since
	$$
	\rob(w,\nu,\theta_1<\theta_2)=\sem{\theta_2}_{w,\nu}-\sem{\theta_1}_{w,\nu}
	=\sem{\theta_2-\theta_1}_{w,\nu}.
	$$
	
	Let $\phi\equiv \neg\phi'$ and let $(\P',v')=\textsc{FormulaRobust}(\phi',P_\dom)$.
	By the induction hypothesis, for every $\nu\models P_\dom$,
	$\rob(w,\nu,\phi')$ is defined iff there exists a unique $b\in \RR\cup\{-\infty,+\infty\}$ such that
	$\nu[v'\gets b]\models\bigcup\P'$, and then $b=\rob(w,\nu,\phi')$.
	The algorithm introduces a fresh output variable $v$, conjoins the constraint $v=-v'$, and projects out $v'$.
	Hence, for every $\nu\models P_\dom$, there exists $a$ with $\nu[v\gets a]\models\bigcup\P$ iff there exists $b$ with
	$\nu[v'\gets b]\models\bigcup\P'$ and $a=-b$.
	Therefore the output is defined exactly when $\rob(w,\nu,\phi')$ is defined, and whenever defined it satisfies
	$$
	a=-\rob(w,\nu,\phi')=\rob(w,\nu,\neg\phi').
	$$
	
	Let $\phi\equiv \phi_1\lor \phi_2$ and let $(\P_i,v_i)=\textsc{FormulaRobust}(\phi_i,P_\dom)$ for $i\in\{1,2\}$.
	By the induction hypothesis, for every $\nu\models P_\dom$ and each $i\in\{1,2\}$,
	$\rob(w,\nu,\phi_i)$ is defined iff there exists a unique $a_i\in \RR\cup\{-\infty,+\infty\}$ such that
	$\nu[v_i\gets a_i]\models\bigcup\P_i$, and then $a_i=\rob(w,\nu,\phi_i)$.
	Under our well-definedness convention (\Cref{rem:well-definedness}), $\rob(w,\nu,\phi_1\lor\phi_2)$ is defined iff both
	$\rob(w,\nu,\phi_1)$ and $\rob(w,\nu,\phi_2)$ are defined.
	The algorithm enforces this by intersecting $\P_1$ and $\P_2$, splitting into the cases
	$v_1\ge v_2$ and $v_1<v_2$, setting $v=v_1$ and $v=v_2$ respectively, and projecting out $v_1,v_2$.
	Therefore, for every $\nu\models P_\dom$, the output is defined iff both $a_1$ and $a_2$ exist, and whenever defined the unique output value satisfies
	$$
	a=\max\{a_1,a_2\}
	=\max\{\rob(w,\nu,\phi_1),\rob(w,\nu,\phi_2)\}
	=\rob(w,\nu,\phi_1\lor\phi_2).
	$$
	
	Let $\phi\equiv \exists x\in J.\,\phi'$, where $x$ is either a value variable (so $J=R$) or a time variable (so $J=I$).
	The algorithm first computes
	$(\P',v')=\textsc{FormulaRobust}\bigl(\phi',\,P_\dom\sqcap\{x\in J\}\bigr)$,
	and then returns $(\P,v)=\textsc{EliminateBySup}(\P',v',x,J)$.
	Fix a valuation $\nu$ of the free variables of $\phi$ with $\nu\models P_\dom$, and define the set of admissible instantiations
	$
	D_\nu \coloneqq \{\, b\in J \mid \rob(w,\nu[x\gets b],\phi')\ \text{is defined}\,\}.
	$
	By the quantitative semantics (together with the well-definedness convention for quantifiers),
	$\rob(w,\nu,\exists x\in J.\phi')$ is defined iff $D_\nu\neq\emptyset$, and then
	$
	\rob(w,\nu,\exists x\in J.\phi')=\sup_{b\in D_\nu}\rob(w,\nu[x\gets b],\phi').
	$
	By the induction hypothesis applied to the recursive call on $\phi'$, for every $b\in J$ the following holds:
	if $b\in D_\nu$, then there exists a unique value
	$a_b=\rob(w,\nu[x\gets b],\phi')\in \RR\cup\{-\infty,+\infty\}$
	such that
	$
	\nu[x\gets b,\,v'\gets a_b]\models \bigcup\P'.
	$
	Equivalently, $\bigcup\P'$ is the graph of the partial map $(\nu,b)\mapsto \rob(w,\nu[x\gets b],\phi')$ over those pairs with $b\in D_\nu$.
	Applying \Cref{lem:eliminatebysup-correct} to $\textsc{EliminateBySup}(\P',v',x,J)$ therefore yields that the resulting output is defined at~$\nu$ iff $D_\nu\neq\emptyset$, and whenever defined its unique output value satisfies
	$$
	a
	=\sup_{b\in D_\nu} a_b
	=\sup_{b\in D_\nu}\rob(w,\nu[x\gets b],\phi')
	=\rob(w,\nu,\exists x\in J.\phi').
	$$
\qed\end{proof}

\begin{proof}[of \cref{cl:correctness-monitor}]
	Fix $i\ge 1$ and $t_0\in[(i-1)\Delta,i\Delta)$.
	Let $P_\dom=\{t\in[(i-1)\Delta,i\Delta)\}$, and let $v_i=(\P_i,v)$ be the pair output by \textsc{Monitor} at iteration~$i$, i.e.\ $(\P_i,v)=\textsc{FormulaRobust}(\phi,P_\dom)$.
	
	\emph{Signal availability on the required window.}
	By construction of \textsc{Monitor}, after the garbage-collection step at the end of iteration $i-1$ (vacuously for $i=1$), each list $\P_f$ contains all signal pieces whose right endpoint is at least $(i-1)\Delta-h$.
	After receiving $w_i$, \textsc{Monitor} appends the fresh piecewise-linear constraints for $w_i$.
	Consequently, immediately before calling \textsc{FormulaRobust} at iteration~$i$, each $\P_f$ encodes the graph of $\llbracket f\rrbracket_w$ on (at least) the time interval $[(i-1)\Delta-h,\ i\Delta)$.
	
	\emph{Correctness on the current segment.}
	Because $\phi$ is pSFO and $h$ is a backward horizon for $\phi$, every signal access time term $\tau$ occurring in $\phi$ satisfies
	$t-h \le \tau \le t$
	for every valuation of bound variables consistent with their quantifier bounds.
	Instantiating $t=t_0\in[(i-1)\Delta,i\Delta)$ yields
	$\tau \in [t_0-h,t_0]\ \subseteq\ [(i-1)\Delta-h,i\Delta)$.
	Thus, whenever $\rob(w,t_0,\phi)$ is defined, all signal values needed to evaluate $\phi$ at $t_0$ lie within the stored signal window represented by the current $\P_f$ lists.
	Since also $t_0\models P_\dom$, we can apply \Cref{lem:fr-correct} to the call $\textsc{FormulaRobust}(\phi,P_\dom)$ at time~$t_0$.
	Therefore, $\rob(w,t_0,\phi)$ is defined iff $v_i(t_0)$ is defined, and in the defined case we have $v_i(t_0)=\rob(w,t_0,\phi)$.
	This proves the first two claims of the theorem.
	
	\emph{Dependence only on $[t_0-h,t_0]$.}
	By the same horizon argument, every signal read performed when evaluating $\phi$ at time $t_0$ occurs at some time in $[t_0-h,t_0]$.
	Hence, whenever $v_i(t_0)$ is defined, modifying $w$ outside $[t_0-h,t_0]$ does not change any accessed signal value and therefore cannot change $\rob(w,t_0,\phi)=v_i(t_0)$.
	
	\emph{Soundness of the garbage-collection step.}
	After producing $v_i$, \textsc{Monitor} drops from each $\P_f$ all pieces ending strictly before $i\Delta-h$.
	For any later evaluation time $t\ge i\Delta$, the required history window satisfies $[t-h,t]\subseteq[i\Delta-h,t]$, so no dropped piece can be accessed in any subsequent iteration.
	Thus garbage collection preserves correctness of all future outputs while ensuring bounded memory.
\qed\end{proof}

\end{document}